\documentclass{article}

\usepackage{cite}
\usepackage{multirow}
\usepackage{graphicx}
\usepackage{subfig}

\usepackage[cmex10]{amsmath}
\usepackage{amsthm}
\usepackage{amssymb}
\usepackage{bm}
\usepackage{latexsym}
\usepackage{bbm}
\newtheorem{theorem}{Theorem}
\newtheorem{corollary}[theorem]{Corollary}
\newtheorem{lemma}[theorem]{Lemma}
\newtheorem{remark}[theorem]{Remark}
\newtheorem{assumption}[theorem]{Assumption}
\newtheorem{example}[theorem]{Example}
\newtheorem{definition}[theorem]{Definition}

\newcommand{\E}{\mathrm{E}}

\newcommand{\Cov}{\mathrm{Cov}}
\newcommand{\Cv}{\mathrm{Cv}}

\newcommand{\peak}{\mathrm{peak}}
\newcommand{\dd}{\mathrm{d}}

\newcommand{\one}{\mathbbm{1}}
\newcommand{\eqdist}{\overset{\mathrm{d}}{=}}
\newcommand{\ds}{\displaystyle}

\allowdisplaybreaks

\begin{document}

\title{A General Formula for the Stationary Distribution of the Age
of Information and Its Application to Single-Server Queues\thanks{This
work was supported in part by MEXT/JSPS KAKENHI Grant Numbers
JP16H06914, JP15K00034, JP16K00034,  JP16H02878, and JP25120008.
This paper was presented in part at 2017 IEEE International Symposium on
Information Theory.}%
}

\author{Yoshiaki Inoue\thanks{Y. Inoue and T. Takine are with 
Department of Information and Communications Technology, 
Graduate School of Engineering, Osaka University, Suita 565-0871, Japan 
(e-mail:yoshiaki@comm.eng.osaka-u.ac.jp; takine@comm.eng.osaka-u.ac.jp).},
Hiroyuki Masuyama\thanks{H. Masuyama and T. Tanaka are with
Department of Systems Science, 
Graduate School of Informatics, Kyoto University, Kyoto 606-8501, Japan
(e-mail: masuyama@i.kyoto-u.ac.jp; tt@i.kyoto-u.ac.jp).},
Tetsuya Takine\footnotemark[2], 
and Toshiyuki Tanaka\footnotemark[3]
}%

\maketitle

\begin{abstract}
This paper considers the stationary distribution of the age of
information (AoI) in information update systems.  We first derive a
general formula for the stationary distribution of the AoI, which
holds for a wide class of information update systems. The formula
indicates that the stationary distribution of the AoI is given in
terms of the stationary distributions of the system delay and the peak
AoI.  To demonstrate its applicability and usefulness, we analyze the
AoI in single-server queues with four different service disciplines:
first-come first-served (FCFS), preemptive last-come first-served
(LCFS), and two variants of non-preemptive LCFS service disciplines.
For the FCFS and the preemptive LCFS service disciplines, the GI/GI/1,
M/GI/1, and GI/M/1 queues are considered, and for the non-preemptive
LCFS service disciplines, the M/GI/1 and GI/M/1 queues are
considered. With these results, we further show comparison results
for the mean AoI's in the M/GI/1 and GI/M/1 queues under those service
disciplines.
\end{abstract}


\section{Introduction}\label{sec:intro}

We consider an information update system composed of an information
source equipped with a sensor, a processor (a server), and a monitor.
The state of the information source changes over time, which is
observed by the sensor occasionally. Whenever the state is sensed, the
sensor generates a packet that contains the sensed data and its
time-stamp, and sends the packet to the server. The server processes
the received data, appends the result to log database, and updates
information displayed on the monitor. The age of information (AoI) is
a primary performance measure in information update systems, which is
defined as the length of time elapsed from the time-stamp of
information being displayed on the monitor.

The information update system described above is an abstraction of
various situations where the freshness of data is of interest, e.g.,
status monitoring in manufacturing systems, satellite imagery for
weather report, tracking trends in social networks, and so on.  The
AoI has recently attracted a considerable attention due to its
applicability to a wide range of information and communication
systems. Readers are referred to \cite{Kosta2017} for a detailed
introduction and review of this subject.

Information update systems are usually modeled as queueing systems,
where packets arriving at a queueing system correspond to information
packets. In most previous work on the analysis of the AoI, the
\textit{mean AoI} is of primary concern, which is defined as
the long-run time-average of the AoI. To be more specific, consider 
a stationary, ergodic queueing system, and let $A_t$ ($t \geq 0$)
denote the AoI at time $t$:
\begin{equation}
A_t := t - \eta_t, \quad t \geq 0,
\label{eq:AoI-process-general}
\end{equation}
where $\eta_t$ ($t \geq 0$) denotes the time-stamp of information
being displayed on the monitor at time $t$.  The mean AoI $\E[A]$ is
defined as
\[
\E[A] = \lim_{T \to \infty} 
\frac{1}{T} \int_0^T A_t \dd t,
\]
and under a fairly general setting, $\E[A]$ is given by
\cite{Kaul12-1}
\begin{equation}
\E[A] = \frac{\E[G_n^{\dag} D_n] + \E[(G^{\dag})^2]/2}{\E[G^{\dag}]},
\label{eq:mean-AoI}
\end{equation}
where $\E[G^{\dag}]$ and $\E[(G^{\dag})^2]$ denote the mean and the
second moment of interarrival times and $\E[G_n^{\dag} D_n]$ denotes
the mean product of the interarrival time $G_n^{\dag}$ of the
$(n-1)$st and the $n$th packets and the system delay $D_n$ of the $n$th packet.
This formula has been the starting point in most previous work on the
analysis of the AoI. As stated in \cite[Page 170]{Kosta2017}, however,
the calculation of the mean AoI based on (\ref{eq:mean-AoI}) is
cumbersome because $G_n^{\dag}$ and $D_n$ are dependent in general and
their joint distribution can take a complicated form.

The purpose of this paper is twofold.  The first one is the derivation
of a general formula for the \textit{stationary distribution} $A(x)$
($x \geq 0$) of the AoI in ergodic information update systems, which
is defined as the long-run fraction of time in which the AoI is not
greater than an arbitrarily fixed value $x$:
\[
A(x) = \lim_{T \to \infty} \frac{1}{T}\int_0^T \one_{\{A_t \leq x\}} \dd t,
\]
where $\one_{\{\cdot\}}$ denotes an indicator function. 
Although the mean AoI $\E[A]$ is a primary performance metric, it
alone is not sufficient to characterize the long-run behavior of the
AoI and its related processes. First of all, if the stationary
distribution $A(x)$ of the AoI is available, we can evaluate the
deviation of the AoI from its mean value.  To support our claim
further, we provide two examples, which show that the stationary
distribution of the AoI plays a central role in the analysis of
AoI-related processes. 

\begin{example}
In \cite{Sun17-it}, a non-linear age penalty function
is introduced to expand the concept of the AoI, which is also referred
to as the Cost of Update Delay (CoUD) metric in \cite{Kosta17}. 
For a non-negative and non-decreasing function $f(x)$ ($x \geq 0$)
with $f(0) = 0$, CoUD $C_t$ at time $t$ is defined as $C_t = f(A_t)$
($t \geq 0$). Clearly, the time-average of CoUD is given in terms of
the stationary  distribution $A(x)$ ($x \geq 0$) of the AoI:
\[
\lim_{T \to \infty} \frac{1}{T}\int_0^T C_t \dd t
=
\int_0^{\infty} f(x) \dd A(x).
\]
\end{example}

\begin{example}
Consider a remote estimation of a stationary Wiener process $\{X_t;\,
t \geq 0\}$ via a channel modeled as a stationary queueing system
\cite{Sun17-isit}. We define $\{\hat{X}_t;\, t \geq 0\}$ as an estimator
$\hat{X}_t = X_{\eta_t}$ of $X_t$ (see (\ref{eq:AoI-process-general})).
As shown in \cite{Sun17-isit}, if a sequence of sampling times is
independent of $\{X_t;\, t \geq 0\}$,
\begin{align*}
\E[(X_t - \hat{X}_t)^2] = \E[(\mathrm{Norm}(0, A_t))^2] = \E[A],
\end{align*}
where $\mathrm{Norm}(\mu, \sigma^2)$ denotes a normal random variable
with mean $\mu$ and variance $\sigma^2$.  Therefore, the mean square
error of the estimator $\{\hat{X}_t;\, t \geq 0\}$ is equal to the
mean AoI for the Wiener process.

As an extension of this result, we consider the stationary
distribution of the estimation error.  It is readily verified that the
characteristic function of the estimation error is given by
\[
\E[e^{i\omega(X_t - \hat{X}_t)}] 
=
\E[e^{i\omega \mathrm{Norm}(0, A_t)}]
=
a^*\left(\frac{\omega^2}{2}\right),
\]
where $a^*(s)$ ($s > 0$) denotes the Laplace-Stieltjes transform
(LST) of the stationary distribution of the AoI.
\end{example}

In \cite{Costa16}, a performance metric called the \textit{peak AoI}
is introduced, which is defined as the AoI immediately before
information updates. The formulation of the peak AoI is simpler than
that of the AoI, and it can be used as an alternative metric of
the freshness of data. In particular, one of the primary
motivations of introducing the peak AoI in \cite{Costa16} is that 
one can characterize its stationary distribution using standard
methods in queueing theory; 
the stationary distribution
of the peak AoI is given in terms of that of the system delay.

In this paper, we show that under a fairly general setting, the
stationary distribution $A(x)$ of the AoI is given in terms of those
of the system delay and the peak AoI. As we will see, this formula
holds sample-path-wise, regardless of the service discipline or the
distributions of interarrival and service times.

Therefore, the analysis of the stationary distribution of the AoI in
ergodic queueing systems is reduced to the analysis of those of the
peak AoI and the system delay, which can be analyzed via
standard techniques in queueing theory. 
An important consequence of this observation is that the peak AoI is
not merely an alternative performance metric to the AoI but rather an
essential quantity in elucidating properties of the AoI. Furthermore,
this observation leads to an alternative formula for the mean AoI in
terms of the \emph{second moments} of the peak AoI and the system delay.

The second purpose of this paper is the derivation of various
analytical formulas for the AoI in single-server queues, which
demonstrates the wide applicability of our general formula. More
specifically, we consider the following four service disciplines:
\begin{itemize}
\item[(A)] First-come, first-served (FCFS),
\item[(B)] Preemptive last-come, first-served (LCFS),
\item[(C)] Non-preemptive LCFS with discarding, and
\item[(D)] Non-preemptive LCFS without discarding.

\end{itemize}
Under the FCFS service discipline, all packets are served in order of
arrival, while under the LCFS service disciplines (B)--(D), the newest
packet is given priority. In the preemptive LCFS discipline, newly
arriving packets immediately start receiving their services on
arrival, interrupting the ongoing service (if any).  In the
non-preemptive LCFS service discipline, on the other hand, arriving
packets have to wait until the completion of the ongoing service, and
waiting packets are also overtaken by those which arrive during their
waiting times.

Note that the non-preemptive LCFS service discipline has two variants,
(C) and (D): overtaken packets are discarded in the service discipline
(C), while overtaken packets remain in the system and they are served
eventually in the service discipline (D).  Although (D) yields a
larger AoI than (C), this service discipline is also of interest in
evaluating the logging overhead caused by sending all generated
packets to the database.  Note that for the preemptive LCFS
discipline, the treatment of overtaken packets (i.e., discarding them
or not) does not affect the AoI performance.

Table \ref{table:literature} summarizes known results for the mean AoI in
standard queueing systems. For the analysis of LCFS queues, all the
previous works listed in Table \ref{table:literature} utilize the
memoryless property of exponential interarrival times to simplify the
derivation of the cross-term  $\E[G_n^{\dag} D_n]$ in
(\ref{eq:mean-AoI}). Also, to the best of our knowledge, no closed
formula for the AoI under the service discipline (D) has been reported
in the literature. As we will see, our general formula is readily
applicable to non-Poisson arrival cases and the service discipline (D). 
We note that in addition to those listed in Table
\ref{table:literature}, there are analytical results for the mean AoI
in queueing models with additional features in the literature, e.g.,
queues with packet deadline \cite{Kam18}, multi-class queues
\cite{Yates19}, and priority queues \cite{Yates18}.

\begin{table}[tp]
\centering
\caption{Summary of known results for the mean AoI in standard queueing
systems, where the fourth entry in Kendall's notation indicates 
the capacity of the waiting room, and P and NP stand for preemptive
and non-preemptive.}
\label{table:literature}
\begin{tabular}{ccc}
\hline
Reference & Service discipline & System 
\\ 
\hline
\multirow{3}{5mm}{\cite{Kaul12-1}}
 & FCFS & M/M/1/$\infty$ \\
 & FCFS & M/D/1/$\infty$ \\
 & FCFS & D/M/1/$\infty$ \\
\hline
\multirow{2}{5mm}{\cite{Costa16}}
 & FCFS & M/M/1/0 \\
& FCFS & M/M/1/1 \\
\hline
\cite{Najm17} & FCFS & M/GI/1/0 \\
\hline
\multirow{2}{5mm}{\cite{Kam16}} 
 & FCFS & M/M/2/$\infty$ \\
 & FCFS & M/M/$\infty$/0 \\
\hline
\cite{Kaul12-2} & P-LCFS & M/M/1/0 \\
\hline
\cite{Najm16} & P-LCFS & M/Gamma/1/0 \\
\hline
\cite{Najm17} & P-LCFS & M/GI/1/0 \\
\hline
\cite{Kaul12-2, Costa16} & NP-LCFS (C) & M/M/1/1 \\
\hline
\cite{Najm16} & NP-LCFS (C) & M/Erlang/1/1 \\
\hline
\end{tabular}
\end{table}

Our contribution to the analysis of the AoI in single-server queues is
summarized as follows.
\begin{namelist}
\item[(A):]
For the FCFS GI/GI/1 queue, we show that the stationary distribution
of the AoI is given in terms of the system delay distribution.  We
also derive upper and lower bounds for the mean AoI in the FCFS
GI/GI/1 queue. In addition, we derive explicit formulas for the LST of
the stationary distribution of the AoI in the FCFS M/GI/1 and GI/M/1
queues.

\item[(B):]
For the preemptive LCFS GI/GI/1, M/GI/1, and GI/M/1 queues, we derive
explicit formulas for the LST of the stationary distribution of the
AoI. In addition, for the preemptive M/GI/1 and GI/M/1 queues, we
derive ordering properties of the AoI in terms of the service time
and the interarrival time distributions.

\item[(C) and (D):]
For the non-preemptive M/GI/1 and GI/M/1 queues with and without
discarding, we derive explicit formulas for the LST of the stationary
distribution of the AoI.
\end{namelist}
In Appendix \ref{appendix:summary}, we also present specialized
formulas for the M/M/1, M/D/1, and D/M/1 queues.

\begin{remark}
\label{remark:G-GI}
Throughout this paper, we strictly distinguish between the symbols `G'
and `GI' in Kendall's notation: `GI' represents that interarrival or
service times are independent and identically distributed (i.i.d.)
random variables, while `G' represents that there are no restrictions
on the arrival or service processes.
\end{remark}

Taking the derivative of the LST of the AoI, we can obtain the mean
and higher moments of the AoI. In all of the above models, we provide
formulas for the mean AoI. We also derive formulas for higher moments
of the AoI when they take simple forms. Furthermore, we obtain
comparison results for the mean AoI among the four service disciplines
in the M/GI/1 and GI/M/1 queues. See Table \ref{table:contribution},
which summarizes our results for the AoI in standard queueing systems.

\begin{table}[p]
\centering
\caption{Summary of our results for the AoI in standard queueing
  models, where the fourth entry in Kendall's notation indicates the
  capacity of the waiting room, and P and NP stand for preemptive and
  non-preemptive.}
\label{table:contribution}
\begin{tabular}{ll}
\hline
\multicolumn{1}{c}{Model} & \multicolumn{1}{c}{Results}
\\ 
\hline
\multirow{2}{*}{FCFS GI/GI/1/$\infty$} & 
LST (Lemma \ref{lemma:a_peak-GIG1}); \\ &
bounds for $\E[A]$ (Theorem \ref{thm:A-mean-bound})
\\\hline
FCFS M/GI/1/$\infty$ & 
LST, $\E[A]$, and $\E[A^2]$  (Theorem \ref{corollary:MG1-GM1-LST} (i))
\\\hline
FCFS GI/M/1/$\infty$ & 
LST, $\E[A]$, and $\E[A^2]$  (Theorem \ref{corollary:MG1-GM1-LST} (ii))
\\\hline
\multirow{2}{*}{FCFS M/M/1/$\infty$} & 
LST, dist.\ function, $\E[A]$ and $\E[A^2]$ \\ & 
(Appendix \ref{appendix:FCFS})
\\\hline
FCFS M/D/1/$\infty$ & $\E[A]$ and $\E[A^2]$ (Appendix \ref{appendix:FCFS})
\\\hline
FCFS D/M/1/$\infty$ &$\E[A]$ and $\E[A^2]$ (Appendix \ref{appendix:FCFS})
\\\hline
\multirow{2}{*}{P-LCFS GI/GI/1/0} &
LST and $\E[A]$ (Theorem \ref{theorem:PR-LCFS-GG1-a^*(s)}); \\ &
Decomposition formula (Corollary \ref{corollary:PR-LCFS-GG1-decomposition})
\\\hline
\multirow{2}{*}{P-LCFS M/GI/1/0} &
LST, $\E[A]$, and $\E[A^2]$ (Corollary \ref{corollary:PR-LCFS-summary} (i));\\&
ordering of $\E[A]$ (Corollary \ref{corollary:P-LCFS-order} (i))
\\\hline
\multirow{2}{*}{P-LCFS GI/M/1/0} &
LST, $\E[A]$, and $\E[A^2]$ (Corollary \ref{corollary:PR-LCFS-summary} (ii));\\&
ordering of dist.\ function (Corollary \ref{corollary:P-LCFS-order} (ii))
\\\hline
\multirow{2}{*}{P-LCFS M/M/1/0} & 
LST, dist.\ function, $\E[A]$ and $\E[A^2]$ \\ & 
(Appendix \ref{appendix:PR-LCFS})
\\\hline
P-LCFS M/D/1/0 & $\E[A]$ and $\E[A^2]$ (Appendix \ref{appendix:PR-LCFS})
\\\hline
P-LCFS D/M/1/0 & $\E[A]$ and $\E[A^2]$ (Appendix \ref{appendix:PR-LCFS})
\\\hline
\multirow{2}{*}{NP-LCFS (C) M/GI/1/1} & 
LST (Theorem \ref{theorem:NP-LCFS} (i)); \\ &
$\E[A]$ (Corollary \ref{corollary:NP-LCFS} (i))
\\\hline
\multirow{2}{*}{NP-LCFS (C) GI/M/1/1} & 
LST (Theorem \ref{theorem:NP-LCFS} (ii)); \\ &
$\E[A]$ (Corollary \ref{corollary:NP-LCFS} (ii))
\\\hline
NP-LCFS (C) M/M/1/1 & LST and $\E[A]$ 
(Appendix \ref{subsection:appendix-NP-LCFS-C})
\\\hline
NP-LCFS (C) M/D/1/1 & $\E[A]$ (Appendix \ref{subsection:appendix-NP-LCFS-C})
\\\hline
NP-LCFS (C) D/M/1/1 & $\E[A]$ (Appendix \ref{subsection:appendix-NP-LCFS-C})
\\\hline
\multirow{2}{*}{NP-LCFS (D) M/GI/1/$\infty$} & 
LST (Theorem \ref{theorem:NP-LCFS} (iii)); \\ &
$\E[A]$ (Corollary \ref{corollary:NP-LCFS} (iii))
\\\hline
\multirow{2}{*}{NP-LCFS (D) GI/M/1/$\infty$} & 
LST (Theorem \ref{theorem:NP-LCFS} (iv)); \\&
$\E[A]$ (Corollary \ref{corollary:NP-LCFS} (iv))
\\\hline
NP-LCFS (D) M/M/1/$\infty$ & LST, $\E[A]$, and $\E[A^2]$ 
(Appendix \ref{subsection:appendix-NP-LCFS-D})
\\\hline
NP-LCFS (D) M/D/1/$\infty$ & $\E[A]$
(Appendix \ref{subsection:appendix-NP-LCFS-D})
\\\hline
NP-LCFS (D) D/M/1/$\infty$ & $\E[A]$ 
(Appendix \ref{subsection:appendix-NP-LCFS-D})
\\\hline\hline
\multirow{3}{*}{M/GI/1} & ordering of $\E[A]$ among FCFS, P-LCFS, \\ &
NP-LCFS (C), and NP-LCFS (D) \\ &
(Theorems \ref{theorem:M/G/1-G/M/1-order} (i) and \ref{theorem:MG1-FCFS<P-LCFS})
\\\hline
\multirow{3}{*}{GI/M/1} & ordering of $\E[A]$ among FCFS, P-LCFS, \\ &
NP-LCFS (C), and NP-LCFS (D) \\ &
(Theorem \ref{theorem:M/G/1-G/M/1-order} (ii))
\\\hline
\end{tabular}
\end{table}

The rest of this paper is organized as follows.  In Section
\ref{sec:sample-path}, we derive a general formula for sample-paths of
the AoI, and using it, we obtain various formulas for the AoI in a
general information update system. In Section \ref{sec:applications},
we consider the applications of the general formula to single-server
queues operated under the FCFS, preemptive LCFS, and non-preemptive
LCFS service disciplines. Furthermore, we  provide some comparison
results for the mean AoI among these  service disciplines in the
M/GI/1 and GI/M/1 queues. Finally, the conclusion of this paper is
provided in Section \ref{sec:conclusion}.

\section{Sample-Path Analysis in a General Setting}
\label{sec:sample-path}

We consider a sample path of the AoI process $\{A_t\}_{t \geq 0}$,
where $A_t$ is defined in (\ref{eq:AoI-process-general}). Note that
the process $\{\eta_t\}_{t \geq 0}$ of the time-stamp of the
information being displayed on the monitor is a step function of $t$,
i.e., it is piece-wise constant and has discontinuous points at which
information is updated. Therefore, any sample path of the AoI process
$\{A_t\}_{t \geq 0}$ is piece-wise linear with slope one and it has
(downward) jumps when information is updated. In what follows, we
assume that $\{A_t\}_{t \geq 0}$ is right-continuous, i.e., $\lim_{t
\to t_0+} A_t = A_{t_0}$.

Any sample path of the AoI process $\{A_t\}_{t \geq 0}$ can be
constructed in the following way. Let $\beta_n$ ($n=1,2,\ldots$)
denote the time instant of the $n$th information update, and 
let $X_n := A_{\beta_n}$ ($n=1,2,\ldots$) denote the AoI immediately
after the $n$th update. Also, let $\beta_0 := 0$ and $X_0 := A_0$.
For simplicity, we assume $\beta_{n-1} < \beta_n$ ($n=1,2,\ldots$). In
these settings, $A_t$ is given by
\begin{equation}
A_t = X_{n-1} + (t - \beta_{n-1}),
\quad
t \in [\beta_{n-1}, \beta_n),
\,
n = 1,2,\ldots.
\label{defn-Z_t}
\end{equation}
The sample-path of the AoI process $\{A_t\}_{t \geq 0}$ is thus
determined completely by the \textit{deterministic marked point
process} $\{(\beta_n,X_n)\}_{n=0,1,\ldots}$. We define $A_{\peak,n}$
($n=1,2,\ldots$) as the $n$th peak AoI (i.e., the AoI immediately
before the $n$th update).
\[
A_{\peak,n} = \lim_{t \to \beta_n-} A_t = X_{n-1} + (\beta_n-\beta_{n-1}).
\]
Figure \ref{fig:AoI-sample-path} shows an example of
sample paths of the AoI process $\{A_t\}_{t \geq 0}$.

\begin{figure}[t]
\centering
\setlength{\unitlength}{1.5mm}%
\begin{picture}(78,33)(-3,-3)
\put(0,0){\vector(0,1){30}}
\put(0,0){\vector(1,0){75}}
\put(0,3){\circle*{1.00}}
\put(0,3){\line(1,1){14.586}}
\put(15,18){\circle{1.00}}
\put(15,7){\circle*{1.00}}
\put(15,7){\line(1,1){19.586}}
\put(35,27){\circle{1.00}}
\put(35,17){\circle*{1.00}}
\put(35,17){\line(1,1){4.586}}
\put(40,22){\circle{1.00}}
\put(40,9){\circle*{1.00}}
\put(40,9){\line(1,1){9.586}}
\put(50,19){\circle{1.00}}
\put(50,2){\circle*{1.00}}
\put(50,2){\line(1,1){19.586}}
\put(70,22){\circle{1.00}}
\put(70,5){\circle*{1.00}}
\put(70,5){\line(1,1){5}}
\put(-3,28){\makebox(0,0){${A_t}$}}
\put(75,-3){\makebox(0,0){${t}$}}
\put(0,-3){\makebox(0,0){${0}$}}
\put(15,-3){\makebox(0,0){${\beta_1}$}}
\put(35,-3){\makebox(0,0){${\beta_2}$}}
\put(40,-3){\makebox(0,0){${\beta_3}$}}
\put(50,-3){\makebox(0,0){${\beta_4}$}}
\put(70,-3){\makebox(0,0){${\beta_5}$}}
\qbezier[16](15,0.5)(15,8.75)(15,17)
\qbezier[25](35,0.5)(35,13.25)(35,26)
\qbezier[20](40,0.5)(40,10.75)(40,21)
\qbezier[17](50,0.5)(50,9.25)(50,18)
\qbezier[15](70,0.5)(70,10.75)(70,21)
\put(15,20){\makebox(0,0){${A_{\peak,1}}$}}
\put(34,29){\makebox(0,0){${A_{\peak,2}}$}}
\put(41,24){\makebox(0,0){${A_{\peak,3}}$}}
\put(50,21){\makebox(0,0){${A_{\peak,4}}$}}
\put(70,24){\makebox(0,0){${A_{\peak,5}}$}}
\put(3,2){\makebox(0,0){${X_0}$}}
\put(13,6){\makebox(0,0){${X_1}$}}
\put(33,16){\makebox(0,0){${X_2}$}}
\put(38,8){\makebox(0,0){${X_3}$}}
\put(48,4){\makebox(0,0){${X_4}$}}
\put(68,4){\makebox(0,0){${X_5}$}}
%
\end{picture}
\caption{A sample path of the AoI process.}
\label{fig:AoI-sample-path}
\end{figure}
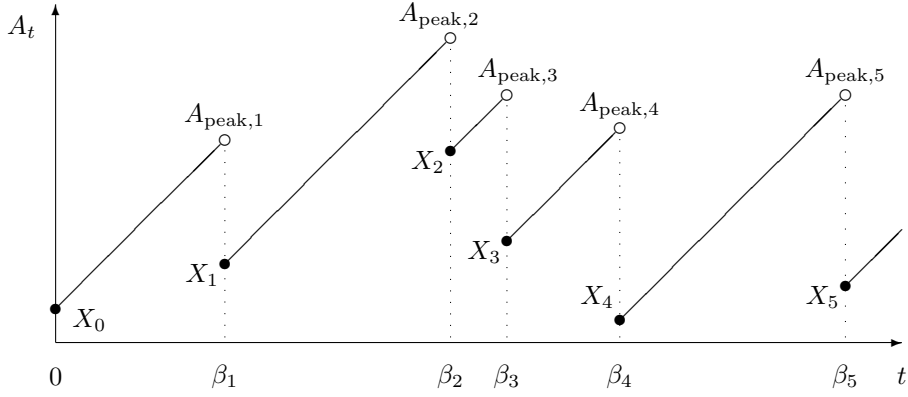

\begin{remark}
As we will see, $X_n$ equals to the system delay of a packet in a
queueing system when the time-stamps of packets are equal to their
arrival times. At this moment, however, we do not make any assumption
on $X_n$ other than its non-negativity, so that the general formula
for the AoI distribution to be obtained is applicable to various
situations, e.g., a network of queues, time-stamps with noise or
uncertainty, information updates obtained by sensor fusion techniques,
and so on.
\end{remark}

\noindent
In what follows, we first show a general formula for sample paths of
the AoI process $\{A_t\}_{t \geq 0}$, which is represented as 
\textit{a deterministic function} in terms of the deterministic
marked point process $\{(\beta_n,X_n)\}_{n=0,1,\ldots}$.
We then discuss the AoI in a general queueing system 
represented as \textit{a stationary, ergodic stochastic model},
assuming that the time-stamps of packets are equal to their arrival
times.

Suppose that the deterministic marked point process
$\{(\beta_n,X_n)\}_{n=0,1,\ldots}$ is given. Let $A^{\sharp}(x)$,
$A_{\peak}^{\sharp}(x)$, and $X^{\sharp}(x)$ ($x \geq 0$) denote the
{\it asymptotic frequency distributions} (see, e.g., \cite[Section
2.6]{El-T99}) of $\{A_t\}_{t \geq 0}$,
$\{A_{\peak,n}\}_{n=1,2,\ldots}$, and $\{X_n\}_{n=1,2,\ldots}$,
respectively, on the sample path:
\begin{alignat}{2}
A^{\sharp}(x) &= 
\lim_{T \to \infty}
\frac{1}{T}\int_0^T \one_{\{A_t \leq x\}} \dd t,
&
x &\geq 0,
\label{eq:A-ave-def}
\\
A_{\peak}^{\sharp}(x)
&= 
\lim_{N \to \infty}
\frac{1}{N} \sum_{n=1}^N \one_{\{A_{\peak,n} \leq x\}},
\qquad
&
x &\geq 0,
\label{eq:A-peak-ave-def}
\\
X^{\sharp}(x)
&= 
\lim_{N \to \infty}
\frac{1}{N} \sum_{n=1}^N \one_{\{X_n \leq x\}},
&
x &\geq 0,
\label{eq:X-ave-def}
\end{alignat}
if the limits exist.  Let $M_t$ ($t \geq 0$) denote
the total number of information updates by time $t$.
\begin{equation}
M_t = \sup\{n \in \{0,1,2,\dots\}; \beta_n \leq t\}.
\label{defn-M_t}
\end{equation}
The following lemma is a sample-path analogue of the elementary
renewal theorem.

\begin{lemma}[{\cite[Lemma 1.1]{El-T99}}]\label{lemma:elementary-renewal}
For $\lambda^{\dag} \in (0,\infty)$, 
\[
\lim_{t \to \infty} \frac{M_t}{t} = \lambda^{\dag}\
\Leftrightarrow\ 
\lim_{n \to \infty} \frac{\beta_n}{n} = \frac{1}{\lambda^{\dag}}.
\]
\end{lemma}

To proceed further, we make the following assumptions.

\begin{assumption}
\label{asm:Y-mean-converge}
\begin{itemize}
\item[(i)]
$\lim_{n \to \infty} (\beta_n/n) = 1/\lambda^{\dag}$ for some
$\lambda^{\dag} \in (0,\infty)$.

\item[(ii)]
The limits in (\ref{eq:A-peak-ave-def}) and (\ref{eq:X-ave-def}) exist 
for each $x \geq 0$.
\end{itemize}
\end{assumption}

Assumption~\ref{asm:Y-mean-converge} (i) implies that $\lim_{n \to \infty}
\beta_n = \infty$, so that $A_t$ is well-defined for $t \in [0,\infty)$.

\begin{lemma}\label{lemma:Z_t-relation}
Under Assumption \ref{asm:Y-mean-converge}, 
the limit (\ref{eq:A-ave-def}) 
exists for each $x \geq 0$ and it is given by
\[
A^{\sharp}(x) 
= 
\lambda^{\dag} \int_0^x (X^{\sharp}(y) - A_{\peak}^{\sharp}(y))\dd y.
\]
\end{lemma}

\begin{proof}
By definition, we have for $T > \beta_1$,
\begin{equation}
\frac{1}{T}\int_0^T \one_{\{A_t \leq x\}} \dd t 
= 
\frac{1}{T}\sum_{n=1}^{M_T-1} S_n(x)
+ \frac{\epsilon(x;T)}{T},
\label{eqn-180103-01}
\end{equation}
where $S_n(x)$ ($n = 1,2,\ldots$) and $\epsilon(x;T)$ are given by
\begin{align}
S_n(x) &=
\int_{\beta_n}^{\beta_{n+1}} \one_{\{A_t \leq x\}} \dd t,
\label{eq:S_n(x)-def}
\\
\epsilon(x;T) &= 
\int_0^{\beta_1} \one_{\{A_t \leq x\}} \dd t
+
\int_{\beta_{M_T}}^T \one_{\{A_t \leq x\}} \dd t.
\label{eq:epsilon}
\end{align}
Thus, to prove the theorem, it suffices to show
that
\begin{align}
\lim_{T \to \infty} \frac{1}{T}
\sum_{n=1}^{M_T-1} S_n(x)
&= 
\lambda^{\dag}
\int_0^x (X^{\sharp}(u) - A_{\peak}^{\sharp}(u)) \dd u,
\label{lim-sum_S_n(x)/T}
\\
\lim_{T \to \infty} \frac{\epsilon(x;T)}{T} &= 0.
\label{lim-ep/T}
\end{align}

We first consider (\ref{lim-ep/T}). From Lemma
\ref{lemma:elementary-renewal} and Assumption
\ref{asm:Y-mean-converge} (i), we have
\begin{equation}
\lim_{T \to \infty}\frac{M_T}{T} = \lambda^{\dag},
\label{eq:M_T/T-limit}
\end{equation}
and thus
\begin{equation}
\lim_{T \to \infty}\frac{\beta_{M_T}}{T}
=
\lim_{T \to \infty}
\left(
\frac{M_T}{T}
\cdot
\frac{\beta_{M_T}}{M_T}
\right)
=
\lim_{T \to \infty}
\frac{M_T}{T}
\cdot
\lim_{N \to \infty} 
\frac{\beta_N}{N}
= 1.
\label{eq:Y_KT-KT}
\end{equation}
It then follows from Assumption \ref{asm:Y-mean-converge} (i),
(\ref{eq:epsilon}), (\ref{eq:M_T/T-limit}), and (\ref{eq:Y_KT-KT}) that
\begin{equation}
0 
\leq
\lim_{T \to \infty} \frac{\epsilon(x;T)}{T} 
\leq
\lim_{T \to \infty} 
\left(
\frac{\beta_1}{T}
+
\frac{\beta_{M_T+1}-\beta_{M_T}}{T} 
\right)
=
\lim_{T \to \infty} \frac{M_T+1}{T} \cdot \frac{\beta_{M_T+1}}{M_T+1} 
-
\lim_{T \to \infty} \frac{\beta_{M_T}}{T} 
= 0,
\end{equation}
which proves (\ref{lim-ep/T}). 

Next we consider (\ref{lim-sum_S_n(x)/T}).
Substituting (\ref{defn-Z_t}) into (\ref{eq:S_n(x)-def}) yields
\begin{align} 
S_n(x)
&=
\int_{\beta_n}^{\beta_{n+1}} \one_{\{X_n + (t-\beta_n) \leq x\}} \dd t
\nonumber
\\
&=
\int_{X_n}^{A_{\peak,n+1}} \one_{\{u \leq x\}} \dd u
\nonumber
\\
&= \int_0^{\infty} 
\one_{\{u \leq x\}} \one_{\{X_n \leq u\}} \one_{\{A_{\peak,n+1} \geq u\}} \dd u
\nonumber
\\
&= 
\int_0^x (\one_{\{X_n \leq u\}} - \one_{\{A_{\peak,n+1} \leq u\}}) \dd u,
\label{eqn-S_n(x)}
\end{align}
where we used 
$\one_{\{A_{\peak,n+1} \leq u\}} = \one_{\{X_n \leq u\}}\one_{\{A_{\peak,n+1} \leq u\}}$
in the last equality. Applying the bounded convergence theorem to (\ref{eqn-S_n(x)}), and
using (\ref{eq:A-peak-ave-def}) and (\ref{eq:X-ave-def}), we obtain
\begin{align*}
\lim_{N \to \infty} \frac{1}{N} \sum_{n=1}^N S_n(x)
&=
\int_0^x (X^{\sharp}(u) - A_{\peak}^{\sharp}(u)) \dd u.
\end{align*}
Combining this with (\ref{eq:M_T/T-limit}) yields
\begin{align*}
\lim_{T \to \infty} \frac{1}{T} \sum_{n=1}^{M_T-1} S_n(x)
&=
\lim_{T \to \infty} \frac{M_T}{T} 
\cdot 
\frac{M_T-1}{M_T} 
\cdot
\frac{1}{M_T-1} 
\sum_{n=1}^{M_T-1} S_n(x)
\nonumber
\\
&= 
\lambda^{\dag} 
\int_0^x (X^{\sharp}(u) - A_{\peak}^{\sharp}(u)) \dd u,
\end{align*}
which completes the proof.
\end{proof}

\begin{remark}
In the proof of Lemma \ref{lemma:Z_t-relation}, we do not use the
inequality $X_n \leq A_{\peak,n}$ ($n=1,2,\ldots$) that should hold in
the context of the AoI.
\end{remark}

In the rest of this paper, we focus on queueing systems 
represented as stationary and ergodic stochastic models, where
$\{(\beta_n, X_n)\}_{n=0,1,\ldots}$ is represented as \textit{a 
random marked point process}. We further restrict our attention to the
case that time-stamps of packets are identical to their arrival times. 
To make the notation simpler, we use the following
conventions. For a non-negative random variable $Y$, let $Y(x)$ ($x
\geq 0$) denote the probability distribution function (PDF) of $Y$ and
let $y^*(s)$ ($s > 0$) denote the LST of $Y$:
\[
Y(x)=\Pr(Y \leq x),
\quad
y^*(s)=\E[\exp(-s Y)].
\]
Furthermore, let $y^{(n)}(s)$ ($s > 0$, $n = 1,2,\ldots$) denote the
$n$th derivative of $y^*(s)$:
\begin{equation}
y^{(n)}(s) 
= 
\frac{\dd^n}{\dd s^n} y^*(s)
= 
(-1)^n \E[Y^n \exp(-s Y)].
\label{eq:x^n-def}
\end{equation}
Note that $Y$ has a finite $n$th moment if and only if $\lim_{s \to 0+}
|y^{(n)}(s)| < \infty$, so that we can write $\E[Y^n] = \lim_{s \to 0+}
(-1)^n y^{(n)}(s)$ even when $\lim_{s \to 0+}(-1)^n y^{(n)}(s)$ diverges
\cite[Page 435]{Feller1971}.
For simplicity, however, we hereafter assume that all random variables
under consideration have finite moments whenever results on their 
moments are stated.

Without loss of generality, we can restrict our attention to a general
FIFO (first-in first-out) queueing system, owing to the following
observation.  In general, arriving packets are classified into two types: 
\textit{informative} and \textit{non-informative} packets. Informative
packets update information being displayed on the monitor, while
non-informative packets do not. If arriving packets are processed on
an FCFS basis, all arriving packets are informative.  On the other
hand, if the order of processing packets is controllable, it might be
reasonable to give priority to the newest packet, because it is
expected to improve the AoI performance. In such a case, older
(overtaken) packets do not update information on the monitor, so that
they are regarded as non-informative. If we \emph{ignore all
non-informative packets and observe only the stream of informative
packets}, they are processed and depart from the system in a FIFO
manner. 

We thus consider a general stationary, ergodic FIFO queueing system,
where only informative packets are visible. For each sample path,
let $\alpha_n$ ($n = 1,2,\ldots$) denote the arrival time of the $n$th
informative packet, and let $\beta_n$ ($n=1,2,\ldots$) denote the
departure (i.e., information update) time of the $n$th informative packet. 
We assume that $\beta_1 > 0 = \beta_0$, $\alpha_n \leq \alpha_{n+1}$,
and $\alpha_n \leq \beta_n < \beta_{n+1}$ ($n = 1,2,\ldots$). Let
$G_n^{\dag}$ ($n=1,2,\ldots$) denote the interarrival time of the
$(n-1)$st and the $n$th informative packets and let $D_n$ ($n =
1,2,\ldots$) denote the system delay of the $n$th informative packet:
\begin{equation}
G_n^{\dag} = \alpha_n - \alpha_{n-1},
\quad
D_n = \beta_n - \alpha_n.
\label{eq:G_n-D_n-def}
\end{equation}
Under the assumption that the time-stamp of the $n$th informative
packet is equal to its arrival time $\alpha_n$, we have
\begin{alignat}{2}
X_n &= D_n, \quad &n=1,2,\ldots,
\label{eq:X_n=D_n}
\\
A_{\peak,n}
&=
\beta_n - \alpha_{n-1}
= G_n^{\dag} + D_n,
\quad
&n = 1,2,\ldots.
\label{eq:A_peak-by-DA}
\end{alignat}
For convenience, we define $D_0 := X_0 = A_0$.

\begin{assumption}\label{assumption:main}
\begin{itemize}
\item[(i)] The mean arrival rate of informative packets is positive
and finite i.e., 
\begin{equation}
\lim_{T \to \infty} \frac{1}{T}\sum_{n=1}^{\infty}
\one_{\{\alpha_n \leq T\}} 
= 
\lambda^{\dag} 
\in 
(0,\infty),
\label{eq:lambda-condition}
\end{equation}
with probability one.

\item[(ii)]
The system is stable, i.e., the mean departure rate of informative
packets is equal to the mean arrival rate of informative packets. 
More concretely,
\begin{equation}
\lim_{T \to \infty} \frac{M_T}{T} = \lambda^{\dag},
\label{eq:rate-coincide}
\end{equation}
with probability one, where $M_t$ ($t \geq 0$) is defined in
(\ref{defn-M_t}). 

\item[(iii)]
The marked point process $\{(\beta_n, D_n)\}_{n=0,1,\ldots}$ 
is stationary and ergodic. 
\end{itemize}
\end{assumption}

In the rest of this paper, we refer to $\lambda^{\dag}$ as the mean
arrival rate of informative packets. Let $\E[G^{\dag}]$ denote the
mean interarrival time of informative packets.
\[
\E[G^{\dag}] 
= 
\lim_{N \to \infty} \frac{1}{N} \sum_{n=1}^N G_n^{\dag}
=
\lim_{N \to \infty} \frac{\alpha_N}{N}.
\]
It then follows from Lemma \ref{lemma:elementary-renewal} that 
under Assumption \ref{assumption:main} (i) and (ii), we have 
\begin{equation}
\E[G^{\dag}] = \frac{1}{\lambda^{\dag}}.
\label{eq:departure-rate-E[G]}
\end{equation}
Recall that $A^{\sharp}(x)$, $A_{\peak}^{\sharp}(x)$, 
and $X^{\sharp}(x)$ ($x \geq 0$) denote asymptotic frequency
distributions defined on a sample path (see (\ref{eq:A-ave-def}),
(\ref{eq:A-peak-ave-def}), (\ref{eq:X-ave-def}), and (\ref{eq:X_n=D_n})).
Under Assumption \ref{assumption:main} (iii), it follows from
the pointwise ergodic theorem \cite[Page 50]{Baccelli2003} that
the asymptotic frequency distributions $A_{\peak}^{\sharp}(x)$,
$X^{\sharp}(x)$, and $A^{\sharp}(x)$ ($x \geq 0$) exist and the
following equations hold with probability one:
\[
A^{\sharp}(x)=A(x), 
\;\;
A_{\peak}^{\sharp}(x)=A_{\peak}(x), 
\;\;
X^{\sharp}(x)=D(x),
\]
where $A$, $A_{\peak}$, and $D$ denote generic random variables for
the stationary AoI, peak AoI, and system delay, respectively. 
Note that by definition, we have for any $t$ and $n$,
\[
A \eqdist A_t, 
\quad
D \eqdist D_n, 
\quad
A_{\peak} \eqdist A_{\peak,n},
\]
where $\eqdist$ stands for the equality in distribution.

Theorem \ref{theorem:time-average-GIG1} presented below is thus
immediate from Lemma \ref{lemma:Z_t-relation} and basic properties of
LST.

\begin{theorem}\label{theorem:time-average-GIG1}
In the general FIFO queueing system satisfying Assumption \ref{assumption:main}, 
\begin{itemize}
\item[(i)] 
the density function $a(x)$ ($x \geq 0$) of the AoI is given by
\begin{equation}
a(x) = \lambda^{\dag} (D(x) - A_{\peak}(x)),
\label{eq:time-average-GIG1-density}
\end{equation}
\item[(ii)] 
the LST $a^*(s)$ ($s > 0$) of the AoI is given by
\begin{equation}
a^*(s) = \lambda^{\dag} \cdot \frac{d^*(s) - a_{\peak}^*(s)}{s},
\label{eq:time-average-GIG1-LST}
\end{equation}
and
\item[(iii)] 
the $k$th ($k= 1,2,\ldots$) moment of the AoI is given by
\begin{equation}
\E[A^k] 
= 
\lambda^{\dag} \cdot \frac{\E[(A_{\peak})^{k+1}] - \E[D^{k+1}]}{k+1}.
\label{eq:time-average-GIG1-moment}
\end{equation}
\end{itemize}
\end{theorem}

\begin{remark}
Letting $k=1$ in (\ref{eq:time-average-GIG1-moment}), we obtain
\begin{equation}
\E[A] = \lambda^{\dag} \cdot \frac{\E[(A_{\peak})^2] - \E[D^2]}{2}.
\label{eq:E[A]-closed}
\end{equation}
The formula for the mean AoI in (\ref{eq:mean-AoI}) is thus reproduced
from (\ref{eq:A_peak-by-DA}), (\ref{eq:departure-rate-E[G]}), and
(\ref{eq:E[A]-closed}).
\end{remark}

\begin{remark}
(\ref{eq:A_peak-by-DA}) and (\ref{eq:departure-rate-E[G]}) imply
\begin{equation}
\lambda^{\dag} = \frac{1}{\E[A_{\peak}] - \E[D]},
\label{eq:lmd^dag-Ap-D}
\end{equation}
so that $a^*(s)$ in (\ref{eq:time-average-GIG1-LST}) trivially
satisfies  $\lim_{s \to 0+} a^*(s) = 1$.
\end{remark}

\begin{remark}
\label{remark:general-FIFO}

Theorem \ref{theorem:time-average-GIG1} can be applied to information
update systems with any service disciplines: Recall that we introduced
the general FIFO queueing system as a virtual system where only
informative packets are visible, and the original system (with
non-informative packets) is not required to be FIFO.

\end{remark}

In the following section, we present applications of Theorem
\ref{theorem:time-average-GIG1} to single-server queues.
As shown in (\ref{eq:x^n-def}), the mean AoI can also be obtained by
taking the derivative of $a^*(s)$ and letting $s \to 0+$, which we
will use repeatedly in the rest of this paper.

\section{Applications to Single-Server Queues}
\label{sec:applications}

In this section, we present analytical results for single-server
queues operated under four service disciplines: FCFS, preemptive
LCFS, and non-preemptive LCFS with and without discarding.
We start with summarizing symbols used in this section.
Throughout this section, we assume that interarrival times 
of packets, which are possibly informative or non-informative, are
i.i.d. We also assume that service times of packets are i.i.d.
Let $G$ (resp.\ $H$) denote a generic random variable for interarrival
times (resp.\ service times). We define $\lambda := 1/\E[G]$ as the mean
arrival rate of packets, and $\rho := \lambda \E[H]$ as the traffic
intensity. 

Let $\tilde{G}$ denote a generic random variable for
residual interarrival times, i.e., the time to the next arrival from a
randomly chosen time instant. Similarly, let $\tilde{H}$ denote 
a generic random variable for residual service times.
By definition, their LSTs are given by
\[
\tilde{h}^*(s) = \frac{1-h^*(s)}{s\E[H]},
\qquad
\tilde{g}^*(s) = \frac{1-g^*(s)}{s\E[G]},
\qquad
s > 0.
\]

Let $G_n^{\dag}$, $\lambda^{\dag}$, $D_n$, and $A_{\peak,n}$ denote
the interarrival times of the $(n-1)$st and the $n$th informative
packets, the mean arrival rate of informative packets, the system
delay of the $n$th informative packet, and the $n$th peak AoI,
respectively, as defined in the preceding section. In addition, let 
$H_n^{\dag}$ ($n = 1,2,\ldots$) denote the service time of the $n$th
informative packet. We also define 
$G^{\dag}$, $H^{\dag}$, $D$, and $A_{\peak}$ as generic random
variables for $G_n^{\dag}$, $H_n^{\dag}$, $D_n$, and $A_{\peak,n}$, 
respectively.
Note that $G^{\dag}$ does not follow the same distribution as $G$ (in
particular, $\lambda^{\dag} \neq \lambda$) unless the service
discipline is FCFS  (see Remark \ref{remark:general-FIFO}).
Similarly, we will see that $H_n^{\dag}$ follows a different
distribution from $H$ under the preemptive LCFS service discipline. 

To avoid trivialities, we make the following assumption in the rest of
this paper.

\begin{assumption} \label{assumption:aperiodic}
At least one of $G$ and $H$ is non-deterministic, i.e., 
the system is not a D/D/1 queue.
\end{assumption}

The rest of this section is organized as follows.
We analyze FCFS queues in Section \ref{ssec:FCFS},
preemptive LCFS queues in Section \ref{ssec:LCFS-P},
and non-preemptive LCFS queues with and without discarding in Section
\ref{ssec:LCFS-NP}. We then provide comparison results among these
service disciplines in Section \ref{ssec:comparison}.

\subsection{Applications to FCFS Queues} \label{ssec:FCFS}

In this subsection, we consider the stationary FCFS GI/GI/1, M/GI/1,
and GI/M/1 queues. Throughout this subsection, we assume $\rho < 1$,
so that the system is stable. 

\begin{remark} \label{remark:mixing-FCFS}
Under the stability condition $\rho < 1$, the system
delay $\{D_n\}_{n=1,2,\ldots}$ of informative packets 
is shown to be a regenerative process with finite mean 
regeneration time \cite[Chapter X, Proposition 1.3]{Asmussen2003},
where the system delay of an informative packet which finds the system
empty on arrival is a regeneration point.
Under Assumption \ref{assumption:aperiodic}, it is then readily
verified that $\{(\beta_n,D_n)\}_{n=0,1,\ldots}$ is 
mixing \cite[Page 49]{Baccelli2003}, so that it is also ergodic.
Theorem \ref{theorem:time-average-GIG1} is thus applicable to FCFS
single-server queues discussed below.
\end{remark}

Under the FCFS service discipline, all arriving packets are
informative, so that
\begin{equation}
G^{\dag} \eqdist G, 
\qquad
\lambda^{\dag} = \lambda = \frac{1}{\E[G]}.
\label{eq:lmd^dag=lmd-FCFS}
\end{equation}
In addition, we have an alternative formula for the $n$th peak AoI 
$A_{\peak,n} = \beta_n - \alpha_{n-1}$ (cf. (\ref{eq:A_peak-by-DA}))
as depicted in Figure \ref{figure:A-peak}:
\begin{align}
A_{\peak,n} &= 
\begin{cases}
G_n^{\dag} + H_n^{\dag}, & G_n^{\dag} > D_{n-1},
\\
D_{n-1} + H_n^{\dag}, & G_n^{\dag} \leq D_{n-1}
\end{cases}
\label{eq:alternative-peak-case}
\\
&= \max(D_{n-1}, G_n^{\dag}) + H_n^{\dag}.
\label{eq:alternative-peak}
\end{align}
Note here that $G_n^{\dag}$ is independent of $D_{n-1}$. 
From (\ref{eq:lmd^dag=lmd-FCFS}), (\ref{eq:alternative-peak}), and
Theorem \ref{theorem:time-average-GIG1}, we then obtain the
following result.

\begin{figure}[t]
\centering
\subfloat[Case of $G_n^{\dag} > D_{n-1}$.]{
\includegraphics[scale=0.5]{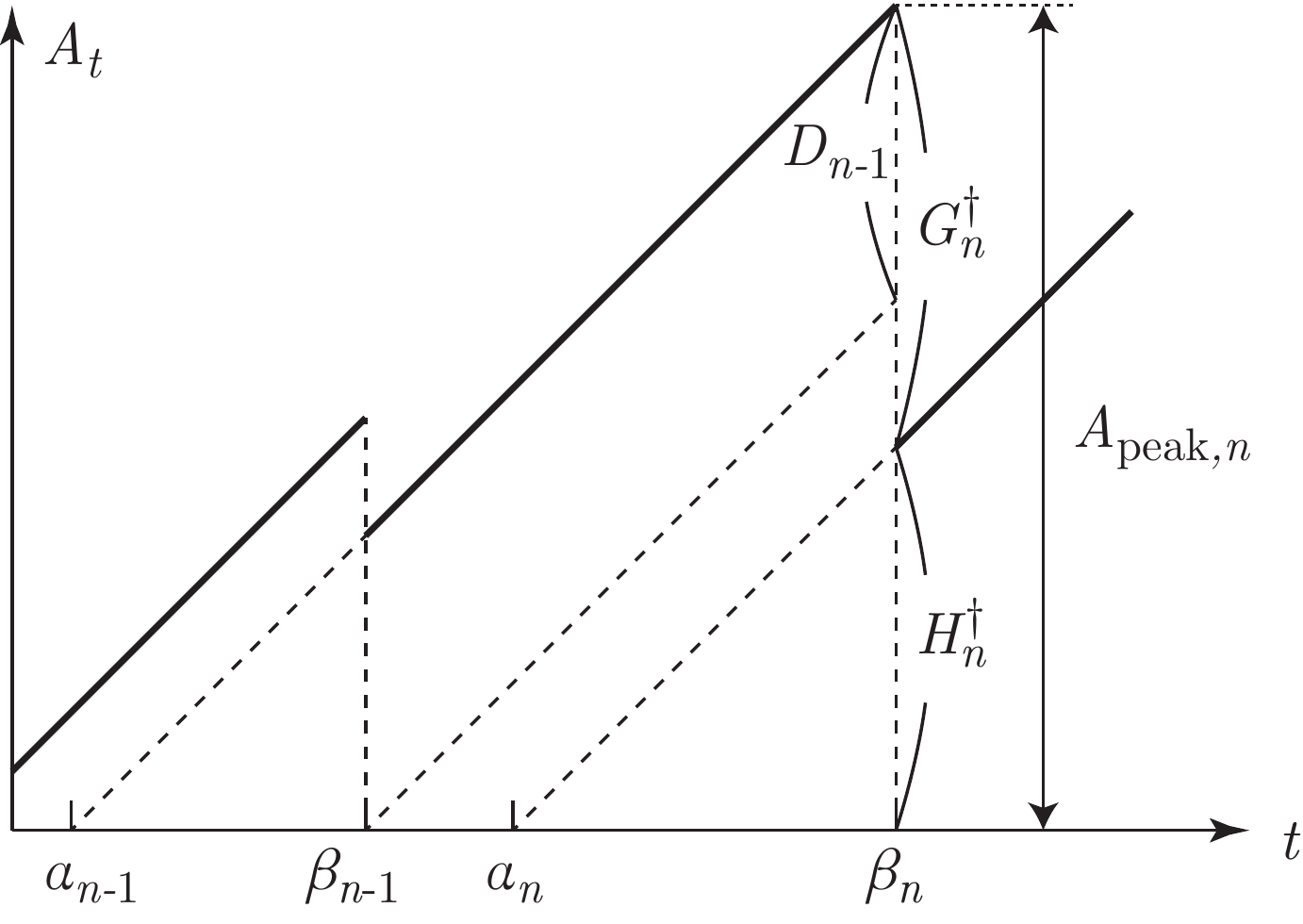}
}
\\[4ex]
\subfloat[Case of $G_n^{\dag} \leq D_{n-1}$.]{
\includegraphics[scale=0.5]{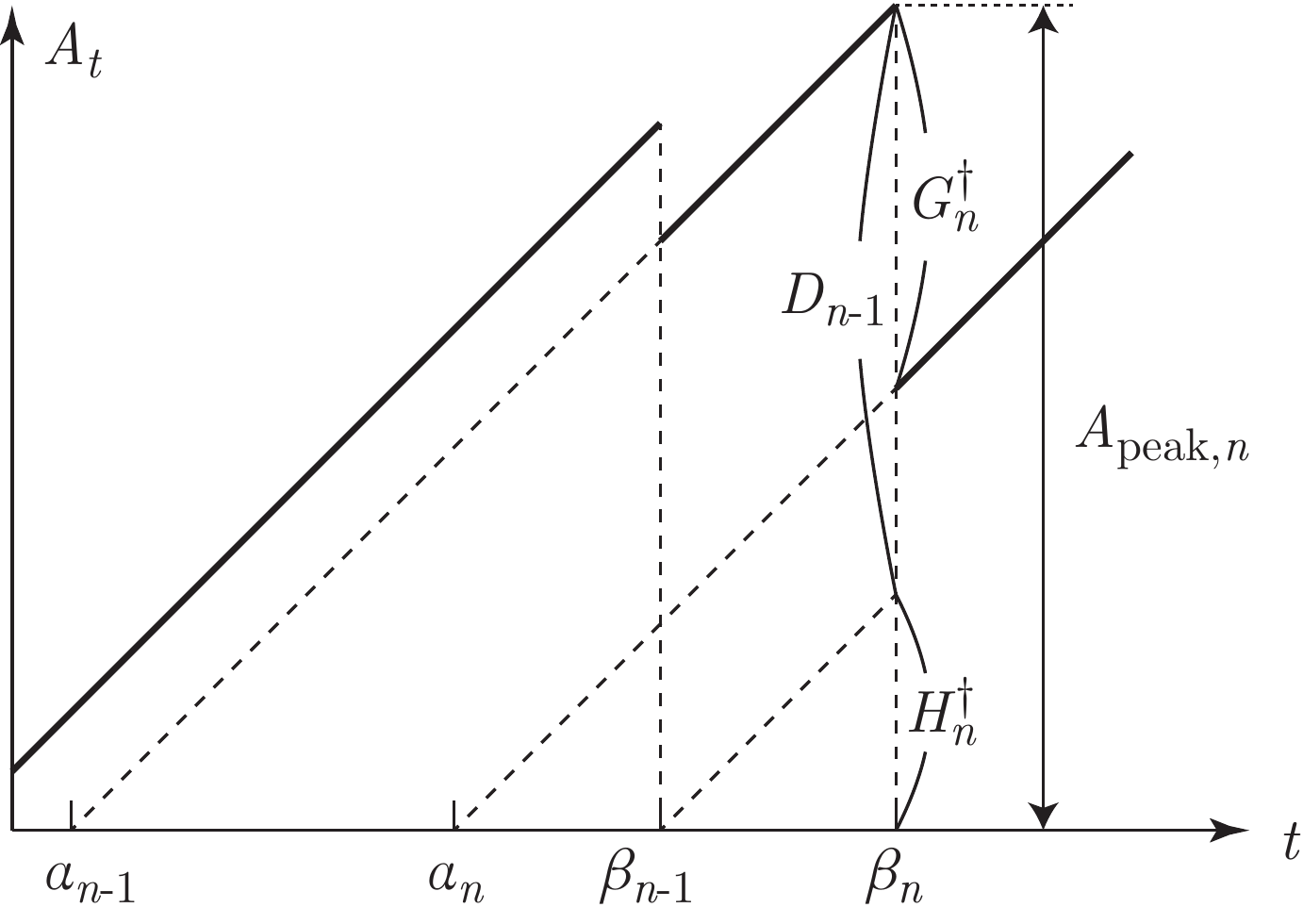}
}
\caption{The relation of the peak AoI $A_{\peak,n}$, the system delay
$D_{n-1}$, the interarrival time $G_n^{\dag}$, and the service time
$H_n^{\dag}$ in the FCFS GI/GI/1 queue.}
\label{figure:A-peak}
\end{figure}

\begin{lemma} \label{lemma:a_peak-GIG1}
In the FCFS GI/GI/1 queue, the LST $a^*(s)$
($s > 0$) of the AoI is given in terms of the system delay
distribution $D(x)$ ($x \geq 0$) by 
\[
a^*(s) = \frac{d^*(s) - a_{\peak}^*(s)}{s\E[G]},
\]
where 
\begin{equation}
a_{\peak}^*(s) 
=
\biggl[
\int_0^{\infty} e^{-s x} G(x) \dd D(x) 
+
\int_0^{\infty} e^{-s x} D(x) \dd G(x)
\nonumber
- \E\left[\one_{\{G=D\}}e^{-sG}\right]
\biggr]
h^*(s).
\label{eq:a_peak-GIG1}
\end{equation}
\end{lemma}

For the FCFS M/GI/1 and GI/M/1 queues, expressions for the system
delay distribution are well-known:
\begin{itemize}
\item[(i)] In the FCFS M/GI/1 queue \cite[Page 199]{Klei75},
\begin{equation}
d^*(s) = \frac{(1-\rho)s}{s-\lambda +\lambda h^*(s)} \cdot h^*(s).
\label{eq:MG1-w_q-w}
\end{equation}
\item[(ii)] In the FCFS GI/M/1 queue \cite[Page 252]{Klei75},
\begin{align}
d^*(s) 
&=
\frac{\mu-\mu\gamma}{s+\mu-\mu\gamma},
\label{eq:d^*(s)-GM1}
\end{align}
where $\gamma$ denotes the unique solution of
\begin{equation}
x=g^*(\mu-\mu x),
\qquad
0 < x < 1.
\label{eq:gamma-def}
\end{equation}
\end{itemize}
Therefore, we can obtain the LST $a^*(s)$ of $A$ by specializing Lemma
\ref{lemma:a_peak-GIG1} with these equations, and 
taking the derivatives of $a^*(s)$, we obtain 
moments of the AoI in these special cases.
\begin{theorem}\label{corollary:MG1-GM1-LST}

\begin{itemize}
\item[(i)] In the FCFS M/GI/1 queue,
\begin{align}
a^*(s) 
&=
\rho d^*(s) \tilde{h}^*(s) 
+ 
d^*(s+\lambda) \cdot \frac{\lambda}{s+\lambda} \cdot h^*(s)
\nonumber
\\
&=
d^*(s) - \frac{(1-\rho) s}{s+\lambda h^*(s+\lambda)} \cdot h^*(s),
\label{eq:LST-M/GI/1}
\\
\E[A] 
&= 
\frac{\lambda \E[H^2]}{2(1-\rho)} + \E[H]
+
\frac{1-\rho}{\rho h^*(\lambda)} \cdot \E[H],
\label{eq:mean-age-MG1}
\\
\E[A^2] 
&=
\frac{2(1-\rho)\left[1 + \rho h^*(\lambda) - \lambda (-h^{(1)}(\lambda))\right]
(\E[H])^2}{(\rho h^*(\lambda))^2}
+ 
\frac{\lambda \E[H^3]}{3(1-\rho)} 
+ 
\frac{(\lambda \E[H^2])^2}{2(1-\rho)^2}
+
\frac{\E[H^2]}{1-\rho}.
\label{eq:2nd-age-MG1}
\end{align}
\item[(ii)] In the FCFS GI/M/1 queue,
\begin{align}
a^*(s) 
& =
\biggl[
\rho d^*(s)
+ \tilde{g}^*(s)-\tilde{g}^*(s+\mu-\mu\gamma) 
\biggr] 
\frac{\mu}{s+\mu},
\nonumber
\\
\E[A]
&=
\frac{\E[G^2]}{2\E[G]} + \E[H]
+
\frac{\rho}{1-\gamma} \Bigl( -g^{(1)}(\mu-\mu\gamma) \Bigr),
\label{eq:mean-age-GM1}
\\
\E[A^2] 
&= 
\frac{\E[G^3]}{3\E[G]} + \rho \E[G^2] + 2(\E[H])^2
+
\frac{\rho}{1-\gamma}\biggl[
g^{(2)}(\mu-\mu\gamma) 
+
2\left(1 + \frac{1}{1-\gamma}\right) 
\Bigl(-g^{(1)}(\mu-\mu\gamma)\Bigr) \E[H] 
\biggr].
\label{eq:2nd-age-GM1}
\end{align}
\end{itemize}
\end{theorem}

In Figure \ref{fig:FCFS_DM1}, the mean $\E[A]$ and the standard
deviation $\mathrm{SD}[A]$ of the AoI in the FCFS
D/M/1 queue with $\E[H]=1$ are plotted as functions of $\rho$.
From this figure, we observe that $\E[A]$ and $\mathrm{SD}[A]$ 
take the minimum values at different values of $\rho$: 
$\E[A]$ is minimal at $\rho = \rho^* \approx 0.516885$,
while $\mathrm{SD}[A]$ is minimal at $\rho \approx 0.408982$.
Figure \ref{fig:FCFS_DM1_E_SD} shows the parametric curve of
$\E[A]$ and $\mathrm{SD}[A]$ with parameter $\rho$.
We observe that given the same mean AoI $\E[A]$, the smaller standard
deviation $\mathrm{SD}[A]$ of the AoI is achieved by the smaller traffic
intensity, which implies that the fewer arrival rate is more effective
than the excessive arrival rate.

\begin{figure}[tbp]
\centering
\subfloat[\mbox{$\E[A]$ and $\mathrm{SD}[A]$ as functions of $\rho$}.]{
\includegraphics[scale=1.0]{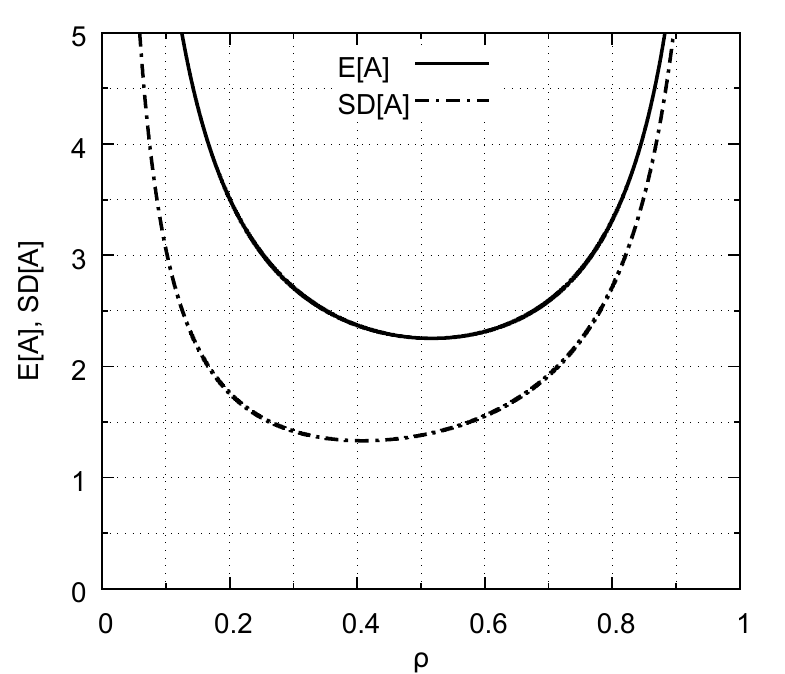}
\label{fig:FCFS_DM1}
}
\subfloat[\mbox{Relation of $\E[A]$ and $\mathrm{SD}[A]$ (parameterized by
$\rho$)}]{
\includegraphics[scale=1.0]{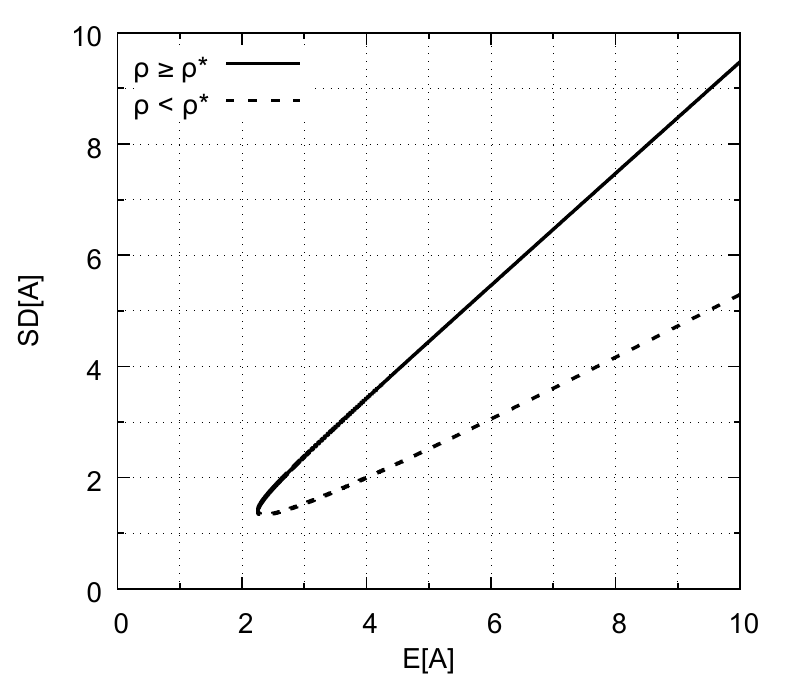}
\label{fig:FCFS_DM1_E_SD}
}
\caption{The mean $\E[A]$ and the standard deviation $\mathrm{SD}[A]$
of the AoI in the FCFS D/M/1 queue ($\E[H]=1$).}
\end{figure}

Before closing this subsection, we derive upper and lower bounds of the
mean AoI $\E[A]$ in the general FCFS GI/GI/1 queue. We rewrite
(\ref{eq:mean-AoI}) to be 
\begin{equation}
\E[A] = \E[D]
+
\frac{1 + (\Cv[G])^2}{2} \cdot \E[G]
+ \frac{\Cov[G_n^{\dag},D_n]}{\E[G]},
\label{eq:A-mean}
\end{equation}
where $\Cv[\cdot]$ denotes the coefficient of variation 
(the standard deviation divided by the mean) and
$\Cov[G_n^{\dag},D_n]$ denotes the covariance of $G_n^{\dag}$ and
$D_n$. Note that $\Cov[G_n^{\dag},D_n]$ does not depend on $n$ because
of the stationarity of the system. 

\begin{lemma}\label{lemma:cov-D-G}
In the FCFS GI/GI/1 queue,
$\Cov[G_n^{\dag},D_n]$ is bounded as follows.
\[
- \E[G]\E[D] \Pr(G < \E[G]) \leq \Cov[G_n^{\dag},D_n] \leq 0.
\]
\end{lemma}
The proof of Lemma \ref{lemma:cov-D-G} is provided in Appendix
\ref{appendix:cov-D-G}.

\begin{theorem}\label{thm:A-mean-bound}
In the FCFS GI/GI/1 queue, the mean AoI is bounded 
as follows.
\begin{align}
\E[A] 
&\geq
\E[D]\Pr(G \geq \E[G]) + \frac{1 + (\Cv[G])^2}{2} \cdot \E[G],
\label{eq:A-mean-lower-bound}
\\
\E[A] 
&\leq 
\E[D] + \frac{1 + (\Cv[G])^2}{2} \cdot \E[G].
\label{eq:A-mean-upper-bound}
\end{align}
\end{theorem}

\begin{proof}
Theorem \ref{thm:A-mean-bound} follows from (\ref{eq:A-mean}) and
Lemma \ref{lemma:cov-D-G}.
\end{proof}

\begin{remark}
The bounds in Theorem \ref{thm:A-mean-bound} are tight in the sense that
both equalities hold in the D/GI/1 queue.
\end{remark}

\begin{remark}
It is known that the mean delay $\E[D]$ in the FCFS GI/GI/1 queue is
bounded by \cite{Dale77,King70}
\begin{align}
\E[D] &\geq \E[H] + \E\left[\{\max(0,H-G)\}^2\right],
\nonumber
\\
\E[D] 
&\leq 
\E[H] 
+ 
\frac{\E[G]}{2(1-\rho)} 
\bigl(\rho(2-\rho)(\Cv[G])^2 + \rho^2 (\Cv[H])^2\bigr).
\label{eq:ED-upper-GG1}
\end{align}
Bounding $\E[D]$ in (\ref{eq:A-mean-lower-bound}) and
(\ref{eq:A-mean-upper-bound}) by these inequalities, we can obtain a
bound for $\E[A]$ in terms of only the distributions of $G$ and $H$.
\end{remark}

Because (\ref{eq:A_peak-by-DA}) implies 
\begin{equation}
\E[A_{\peak}] = \E[D] + \E[G],
\label{eq:E[A_peak]}
\end{equation}
the following corollary is immediate from (\ref{eq:A-mean-upper-bound}).

\begin{corollary}\label{cor:A<A_peak}
In the FCFS GI/GI/1 queue, if $\Cv[G] \leq 1$, then
$\E[A] \leq \E[A_{\peak}]$.
\end{corollary}

\begin{remark}\label{remark-A=<A-peak}
$\E[A] \leq \E[A_{\peak}]$ does not hold in general, which might sound
counterintuitive.  A simple counter example is the case that $\E[H] =
0$ and $\Cv[G] > 1$. Because $\E[H]=0$ leads to $\E[D] = 0$, we have
$\Cov[G_n^{\dag},D_n] = 0$. The inequality $\E[A] > \E[A_{\peak}]$ then
follows from (\ref{eq:A-mean}), (\ref{eq:E[A_peak]}), and $\Cv[G] >
1$. 

Furthermore, we can derive a sufficient condition for 
$\E[A] > \E[A_{\peak}]$ in the stationary FCFS GI/M/1 queue. 
It follows from (\ref{eq:ED-upper-GG1}), (\ref{eq:E[A_peak]}), 
and $\Cv[H]=1$ that 
\[
\E[A_{\peak}] \leq \E[H] + 
\frac{\E[G]}{2(1-\rho)} 
\bigl(\rho(2-\rho)(\Cv[G])^2 + \rho^2\bigr) +\E[G].
\]
On the other hand, from (\ref{eq:mean-age-GM1}) and
$-g^{(1)}(\mu-\mu\gamma)> 0$, we have
\[
\E[A] \geq \frac{\E[G]}{2} \bigl((\Cv[G])^2+1\bigr) + \E[H]. 
\]
Therefore, 
\[
\E[A] - \E[A_{\peak}]
\geq
\frac{\E[G]}{2(1-\rho)}
\left[
(1-3\rho+\rho^2) (\Cv[G])^2 - (1-\rho+\rho^2)\right].
\]
Note here that $1-3\rho+\rho^2 >0$ if $\rho < (3-\sqrt{5})/2$.
We thus conclude that in the stationary FCFS GI/M/1 queue, 
\[
\rho \in (0,(3-\sqrt{5})/2)\ 
\mbox{and}\ 
(\Cv[G])^2 \geq 1 + \frac{2\rho}{1-3\rho+\rho^2}
\ \Rightarrow\ 
\E[A] > \E[A_{\peak}].
\]
\end{remark}

\subsection{Applications to preemptive LCFS queues}
\label{ssec:LCFS-P}

In this subsection, we consider the AoI in the stationary preemptive
LCFS GI/GI/1, M/GI/1, and GI/M/1 queues. As stated in Section
\ref{sec:sample-path}, even in LCFS systems, informative packets
arrive and depart in a FIFO manner, so that Theorem
\ref{theorem:time-average-GIG1} is applicable under Assumption
\ref{assumption:main}. Recall that Assumption \ref{assumption:main} is
posed only on informative packets, and therefore the results in this
subsection may hold even if $\rho \geq 1$.

We first consider a GI/GI/1 queue. We define
$\zeta$ as
\[
\zeta = \Pr(G < H).
\]
Note that if $\zeta = 1$, every service is interrupted by the next
arrival with probability one, so that the AoI goes to infinity as time
passes. On the other hand, if $\zeta=0$, all packets are informative,
i.e., $d^*(s)=h^*(s)$ and $a_{\peak}^*(s) = g^*(s)h^*(s)$. 
To avoid trivialities and for simplicity, we assume the following:
\begin{assumption}\label{assumption:P-LCFS-GI-GI-1}
$\Pr(H=G)=0$ and $0 < \zeta < 1$.
\end{assumption}

\begin{remark}
We can verify that under Assumption \ref{assumption:P-LCFS-GI-GI-1},
the system delay of informative packets $\{D_n\}_{n=1,2,\ldots}$
is a regenerative process with finite mean regeneration time, 
where the system delay of an informative packet which finds the system
empty on arrival is a regeneration point. 
Therefore, under Assumptions \ref{assumption:aperiodic} and
\ref{assumption:P-LCFS-GI-GI-1}, $\{(\beta_n,D_n)\}_{n=0,1,\ldots}$ is
mixing and ergodic (cf.\ Remark \ref{remark:mixing-FCFS}).
Theorem \ref{theorem:time-average-GIG1} is thus applicable to the
preemptive single-server queues discussed below.
\end{remark}

Under the preemptive LCFS service discipline, a packet becomes
non-informative if the next packet arrives before its service
completion. We thus have
\begin{equation}
D_n = H_n^{\dag} \eqdist H_{<G},
\label{eq:D_n-PR-LCFS-GG1}
\end{equation}
where $H_{<G}$ denotes a generic conditional random variable for 
a service time given that it is smaller than interarrival time.

Note that the $m$th ($m=1,2,\ldots$) arriving packet after time
$\beta_n$ becomes the $(n+1)$st informative packet with
probability $\zeta^{m-1} (1-\zeta)$, and the peak AoI in this
case is given by
\begin{equation}
A_{\peak,n+1} 
\eqdist 
G_{>H} 
+ 
\sum_{i=1}^{m-1} G_{<H}^{[i]}
+ H_{<G},
\label{eq:PR-LCFS-GG1-(n+1)st-peak}
\end{equation}
where $G_{>H}$ (resp.\ $G_{<H}^{[i]}$) denotes a generic conditional
random variable for an interarrival time given that it is greater
(resp.\ smaller) than a service time. We can verify that
$G_{>H}$ represents the interarrival time of the first packet
arriving after the departure of the $n$th informative packet,
$G_{<H}^{[i]}$ represents the interarrival time of the $i$th
non-informative packet and the next packet, and $H_{<G}$ represents
the service time of the $(n+1)$st informative packet.
We define $G_{<H} := G_{<H}^{[1]}$.
Note that the LSTs of $H_{<G}$,
$G_{<H}$, and $G_{>H}$ are given by
\begin{align*}
h_{<G}^*(s) 
&=
\frac{1}{1-\zeta} \int_0^{\infty} e^{-s x} (1-G(x))\dd H(x),
\\
g_{<H}^*(s) 
&=
\frac{1}{\zeta} \int_0^{\infty} e^{-s x} (1-H(x))\dd G(x),
\\
g_{>H}^*(s) 
&=
\frac{1}{1-\zeta} \int_0^{\infty} e^{-s x} H(x) \dd G(x).
\end{align*}
By definition, we have
\[
g^*(s) = \zeta g_{<H}^*(s) + (1-\zeta) g_{>H}^*(s).
\]

We can obtain the following result from Theorem
\ref{theorem:time-average-GIG1}, noting that $m+1$ random variables
on the right-hand side of (\ref{eq:PR-LCFS-GG1-(n+1)st-peak})
are mutually independent in the GI/GI/1 queue
(see Appendix \ref{appendix:PR-LCFS-GG1-a^*(s)} for more details).

\begin{theorem}\label{theorem:PR-LCFS-GG1-a^*(s)}
In the preemptive LCFS GI/GI/1 queue,
\begin{align}
a^*(s) 
&=
h_{<G}^*(s) \cdot \tilde{g}^*(s) \cdot \frac{1-\zeta}{1-\zeta g_{<H}^*(s)},
\label{eq:PR-LCFS-GG1-a^*(s)}
\\
\E[A] &= 
\E[H_{<G}] 
+ \frac{E[G^2]}{2\E[G]} 
+ \frac{\zeta}{1-\zeta} \E[G_{<H}].
\label{eq:PR-LCFS-GG1-EA}
\end{align}
\end{theorem}

\begin{corollary}\label{corollary:PR-LCFS-GG1-decomposition}
In the preemptive LCFS GI/GI/1 queue, the AoI is
decomposed stochastically into three independent factors:
\begin{equation}
A \eqdist H_{<G} + \tilde{G} + Z,
\label{eq:PR-LCFS-GG1-decomposition}
\end{equation}
$\tilde{G}$ denotes a generic random variable for residual interarrival times,
and $Z$ denotes a random variable for an interval from the arrival of
a randomly chosen packet to the arrival of the next informative
packet, whose LST is given by
\[
z^*(s) 
= 
\frac{1-\zeta}{1-\zeta g_{<H}^*(s)}.
\]
\end{corollary}

\begin{proof}
(\ref{eq:PR-LCFS-GG1-decomposition}) immediately follows from Theorem
\ref{theorem:PR-LCFS-GG1-a^*(s)}. The probabilistic interpretation of
$Z$ comes from
\[
\frac{1-\zeta}{1-\zeta g_{<H}^*(s)}
=
\sum_{m=1}^{\infty} \zeta^{m-1} (1-\zeta) \cdot (g_{<H}^*(s))^{m-1}.
\qedhere
\]
\end{proof}

Below, we consider two special cases: the preemptive LCFS M/GI/1 and GI/M/1
queues. Note that (\ref{eq:PR-LCFS-GG1-a^*(s)}) in these cases
takes simple forms,
and the moments of the AoI are readily obtained.

\begin{corollary}\label{corollary:PR-LCFS-summary}
\begin{itemize}
\item[(i)] In the preemptive LCFS M/GI/1 queue,
\begin{align}
a^*(s) 
&= 
\frac{\lambda h^*(s+\lambda)}{s+\lambda h^*(s+\lambda)},
\label{eq:theorem-PR-LCFS-AoI}
\\
\E[A] &= \frac{\E[H]}{\rho h^*(\lambda)},
\label{eq:mean-AoI-LCFS-P-MG1}
\\
\E[A^2] 
&= 
2 [1-\lambda (-h^{(1)}(\lambda))]
\left(\frac{\E[H]}{\rho h^*(\lambda)}\right)^2. 
\label{eq:2m-AoI-LCFS-P-MG1}
\end{align}
\item[(ii)]
In the preemptive LCFS GI/M/1 queue,
\begin{align}
a^*(s) 
&= 
\tilde{g}(s) \cdot \frac{\mu}{s+\mu},
\label{eq:theorem-PR-LCFS-GM1-AoI}
\\
\E[A^n] &= \sum_{m=0}^n \frac{n!E[G^{m+1}](\E[H])^{n-m}}{(m+1)!\E[G]},
\;\;
n=1,2,\ldots,
\nonumber
\end{align}
and in particular,
\begin{align}
\E[A] &= \frac{\E[G^2]}{2\E[G]} + \E[H],
\label{eq:mean-AoI-GM1-LCFS-P}
\\
\E[A^2] &= \frac{\E[G^3]}{3\E[G]} + 2\E[H]\E[A].
\nonumber
\end{align}
\end{itemize}
\end{corollary}

\begin{remark}
(\ref{eq:mean-AoI-LCFS-P-MG1}) is identical to Eq.\ (15) in
\cite{Najm17}.
\end{remark}

From Corollary \ref{corollary:PR-LCFS-summary}, we can discuss the
effect of the variability of service and interarrival times on the
AoI. To this end, we introduce the convex order of random variables.

\begin{definition}[{\cite[Page 109]{Shaked06}}]
Consider two random variables $X$ and $Y$ with the same mean
$\E[X]=\E[Y]$. $X$ is said to be smaller than or equal to $Y$
in the convex order (denoted by $X \leq_{\rm cx} Y$) if and only if
\[
\E[\phi(X)] \leq \E[\phi(Y)],
\]
for all convex functions $\phi$, provided the expectations exist.
\end{definition}

By definition, the convex order is a partial
order over the set of all real-valued random variables.
Roughly speaking, $X \leq _{\rm cx} Y$ implies that $Y$ is more
variable than $X$ \cite{Shaked06}. 
In particular, it is readily verified that
\begin{equation}
X \leq_{\rm cx} Y\ \Rightarrow\ \Cv[X] \leq \Cv[Y].
\label{eq:convex-order}
\end{equation}

Consider two preemptive LCFS queues with the same mean interarrival
time $\E[G]$ and the same mean service time $\E[H]$.
Let $G^{\langle k \rangle}$, $H^{\langle k \rangle}$, and $A^{\langle
k \rangle}$ ($k=1,2$) denote generic random variables for 
interarrival times, service times and the AoI in the $k$th queue. 

\begin{corollary}\label{corollary:P-LCFS-order}
\begin{itemize}
\item [(i)] For two preemptive LCFS M/GI/1 queues,
\begin{equation}
H^{\langle 1 \rangle} \leq_{\rm cx} H^{\langle 2 \rangle}\
\Rightarrow\ 
\E[A^{\langle 1 \rangle}] \geq \E[A^{\langle 2 \rangle}].
\label{eq:EA-ordering-PR-LCFS-MG1}
\end{equation}
\item [(ii)] For two preemptive LCFS GI/M/1 queues,
\begin{equation}
G^{\langle 1 \rangle} \leq_{\rm cx} G^{\langle 2 \rangle}\
\Rightarrow\ 
A^{\langle 1 \rangle} \leq_{\rm st} A^{\langle 2 \rangle},
\label{eq:EA-ordering-PR-LCFS-GM1}
\end{equation}
where $\leq_{\rm st}$ stands for the usual stochastic order
\cite[Page 4]{Shaked06}, i.e., 
$
A^{\langle 1 \rangle} \leq_{\rm st} A^{\langle 2 \rangle}
\Leftrightarrow
\E[\phi(A^{\langle 1 \rangle})] \leq \E[\phi(A^{\langle 1 \rangle})]
$
for all non-decreasing functions $\phi$, provided that the
expectations exist.
\end{itemize}
\end{corollary}

\begin{remark}
Clearly, (\ref{eq:EA-ordering-PR-LCFS-GM1}) implies
\[
G^{\langle 1 \rangle} \leq_{\rm cx} G^{\langle 2 \rangle}\
\Rightarrow\ 
\E[A^{\langle 1 \rangle}] \leq \E[A^{\langle 2 \rangle}].
\]
\end{remark}

\begin{proof}
We have (\ref{eq:EA-ordering-PR-LCFS-MG1}) from (\ref{eq:mean-AoI-LCFS-P-MG1})
because $h^*(\lambda) = \E[e^{-\lambda H}]$, and $\exp(-sx)$ ($s > 0$)
is a convex function. 

We then consider (\ref{eq:EA-ordering-PR-LCFS-GM1}).
Note that (\ref{eq:theorem-PR-LCFS-GM1-AoI}) indicates 
\begin{equation}
A \eqdist \tilde{G}+H.
\label{eq:A-decompose-PR-LCFS-GM1}
\end{equation}
Because $H$ is independent of $\tilde{G}^{\langle k\rangle}$
($k = 1,2$), it follows from \cite[Theorem 1.A.3, Eq.\ (3.A.7)]{Shaked06} that
\begin{align*}
G^{\langle 1 \rangle} \leq_{\mathrm{cx}} G^{\langle 2 \rangle}\ 
&\Rightarrow\
\tilde{G}^{\langle 1 \rangle} \leq_{\mathrm{st}} \tilde{G}^{\langle 2 \rangle}
\nonumber
\\
&\Rightarrow\
\tilde{G}^{\langle 1 \rangle} + H \leq_{\mathrm{st}}
\tilde{G}^{\langle 2 \rangle} + H.
\end{align*}
Therefore, we obtain (\ref{eq:EA-ordering-PR-LCFS-GM1})
from (\ref{eq:A-decompose-PR-LCFS-GM1}).
\end{proof}

Corollary \ref{corollary:P-LCFS-order} (i) shows that the preemptive
LCFS service discipline is particularly effective in terms of the mean
AoI when service times are \textit{highly variable},
given that packets arrive according to a Poisson process.
On the other hand, Corollary \ref{corollary:P-LCFS-order} (ii) shows
that lowering the variability of interarrival times reduces 
the AoI, when the service time distribution is exponential.

Note that Jensen's inequality leads to $\E[X] \leq_{\mathrm{cx}} X$ 
for any random variable $X$, i.e., the deterministic distribution
achieves the minimum in the convex order \cite[Theorem 3.A.24]{Shaked06}. 
Therefore, the M/D/1 queue gives \textit{the maximum} (i.e., the
worst) mean AoI among all preemptive LCFS M/GI/1 queues with the same
mean interarrival time $\E[G]$ and the same mean service time $\E[H]$.
On the other hand, the D/M/1 queue gives \textit{the minimum} 
(i.e., the best) mean AoI among all preemptive LCFS GI/M/1 queues
with the same mean interarrival time $\E[G]$ and the same mean service
time $\E[H]$.

We now present some numerical examples, where we fix $\E[H]=1$ and 
consider the following probability distributions for service times: 
\begin{itemize}
\item Deterministic distribution ($\Cv[H] = 0$).
\item Gamma distribution ($0 < \Cv[H] < 1$).
\item Exponential distribution  ($\Cv[H] = 1$).
\item Hyper-exponential distribution of order two with balanced means
\cite[Page 359]{Tijms1994} ($\Cv[H] > 1$).
\end{itemize}
Note that this setting enables us to specify the service time 
distribution uniquely once we do $\Cv[H]$.

In Figure \ref{fig:LCFS_P_MG1}, the mean AoI $\E[A]$ in the preemptive
LCFS M/GI/1 queue, given by (\ref{eq:mean-AoI-LCFS-P-MG1}), is plotted as
a function of $\rho$ for various values of $\Cv[H]$. This figure
confirms Corollary \ref{corollary:P-LCFS-order} (i), and we observe
that the variability of the service times has a significant impact on
the mean AoI. 

Next, using (\ref{eq:PR-LCFS-GG1-EA}), we present Figures
\ref{fig:LCFS_P_DBh1} and \ref{fig:LCFS_P_DGm1} for $\E[A]$ in the
D/GI/1 queue, where $\Cv[H] \geq 1$ in Figure \ref{fig:LCFS_P_DBh1}
and $0 \leq \Cv[H] \leq 1$ in Figure \ref{fig:LCFS_P_DGm1}. 
When $\Cv[H] \geq 1$, $\E[A]$ decreases
with an increase in $\Cv[H]$, as in the case of Poisson arrivals 
in Figure \ref{fig:LCFS_P_MG1}.
When $\Cv[H] \leq 1$, on the other hand, 
the mean AoI $\E[A]$ for $\rho < 1$ \textit{increases} with an
increase in the variability of service times.

Therefore, the preemptive LCFS service discipline is particularly
effective when service times are highly variable. In addition,
if arrival times of packets are deterministic, this service discipline
is also effective for less variable service times.
Finally, we observe from Figures \ref{fig:LCFS_P_MG1}--\ref{fig:LCFS_P_DGm1}
that for gamma service time distributions with $\Cv[H] < 1$, the mean AoI
$\E[A]$ diverges to infinity as the arrival rate goes to infinity,
which is not the case for exponential and hyper-exponential service
time distributions. These phenomena are in fact not dependent on the
details of the service time distributions, but hold in general as
stated in the following theorem.

\begin{theorem} \label{theorem:LCFS-P-limit}
Consider the M/GI/1 and D/GI/1 queues with an absolutely continuous
service time distribution. If the probability density function
$h(x)$ ($x \geq 0$) of service times is bounded and continuous, 
we have for both of the M/GI/1 and D/GI/1 queues,
\begin{equation}
\E[A] \rightarrow \frac{1}{h(0)}
\qquad
(\lambda \to \infty),
\label{eq:LCFS-P-limit}
\end{equation}
in the sense that the limit value of $\E[A]$ is finite if $h(0) \neq 0$, 
and otherwise $\E[A]$ diverges to infinity as $\lambda \to \infty$.
\end{theorem}

\noindent
The proof of Theorem \ref{theorem:LCFS-P-limit} is provided in
Appendix \ref{appendix:LCFS-P-limit}.

\begin{figure}[tbp]
\centering
\includegraphics[scale=1.0]{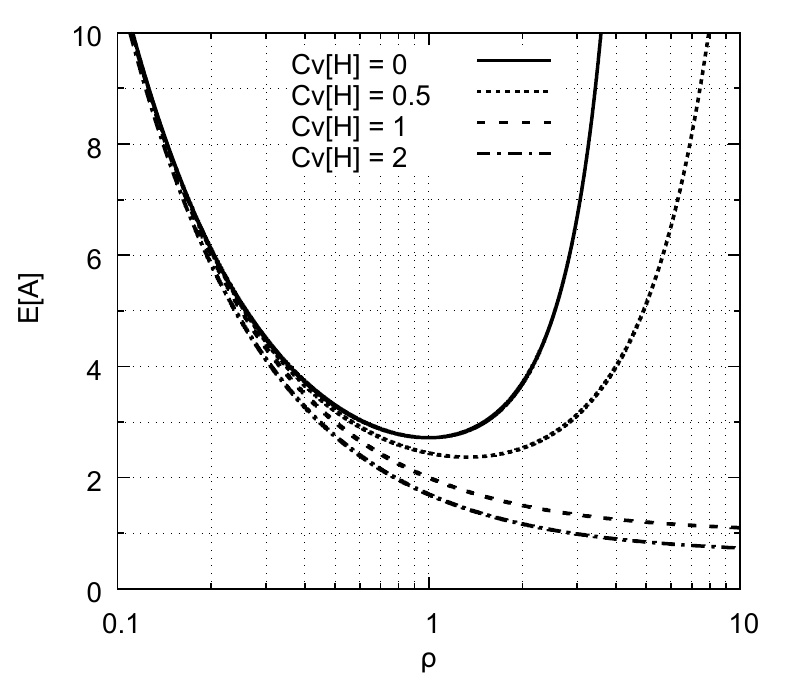}
\vspace{-1ex}
\caption{$\E[A]$ in the preemptive LCFS M/GI/1 queue.}
\label{fig:LCFS_P_MG1}
\mbox{}
\\[1ex]
\includegraphics[scale=1.0]{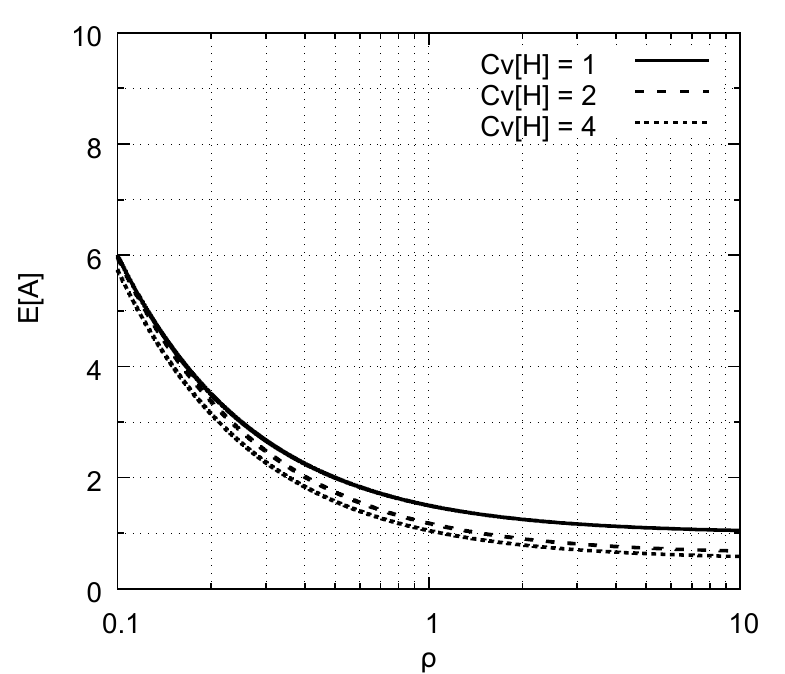}
\vspace{-1ex}
\caption{$\E[A]$ in the preemptive LCFS D/GI/1 queue ($\Cv[H] \geq 1$).}
\label{fig:LCFS_P_DBh1}
\mbox{}
\\[1ex]
\includegraphics[scale=1.0]{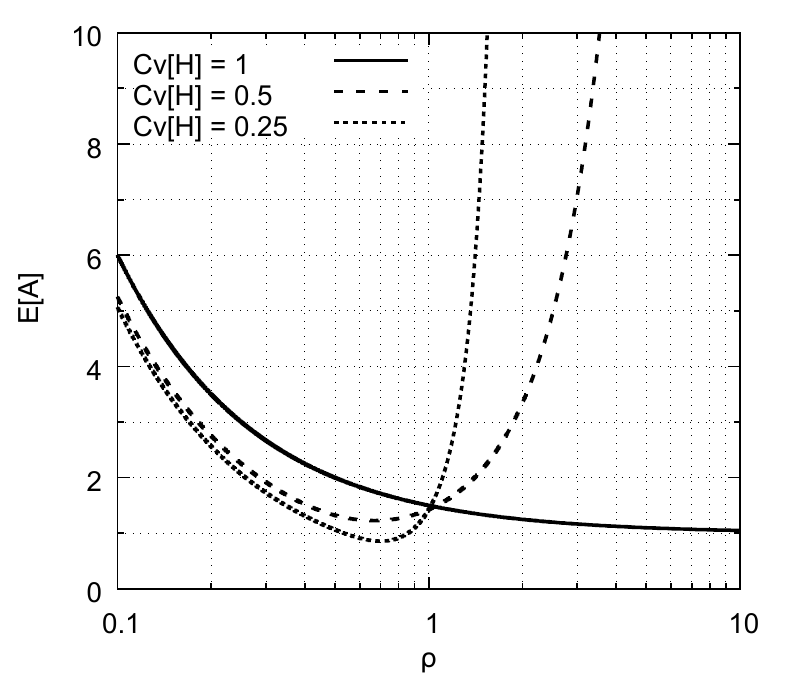}
\vspace{-1ex}
\caption{$\E[A]$ in the preemptive LCFS D/GI/1 queue ($\Cv[H] \leq 1$).}
\label{fig:LCFS_P_DGm1}
\end{figure}

\subsection{Applications to non-preemptive LCFS queues with and without discarding}
\label{ssec:LCFS-NP}

In this subsection, we apply Theorem \ref{theorem:time-average-GIG1} to
the stationary non-preemptive LCFS M/GI/1 and GI/M/1 queues with and without discarding.
In non-preemptive LCFS queues with discarding, non-informative
packets are discarded without receiving their services. In
non-preemptive LCFS queues without discarding, on the other hand, all
packets are served eventually, whether they are informative or not.

Similarly to preemptive LCFS systems, Theorem \ref{theorem:time-average-GIG1} 
is applicable under Assumption \ref{assumption:main}, so that the
results in this subsection may hold even if $\rho \geq 1$. For systems
without discarding, however, we focus on the case of $\rho < 1$
because otherwise the mean number of waiting non-informative packets
goes to infinity as time passes.

\begin{remark}
Because the non-preemptive LCFS service disciplines are
work-conserving, the workload process is identical to that 
of the FCFS case. Similarly to Remark \ref{remark:mixing-FCFS}, we can
thus show that under the stability condition $\rho < 1$, the system
delay $\{D_n\}_{n=1,2,\ldots}$ of informative packets is a
regenerative process with finite mean regeneration time, where the
system delay of an informative packet which finds the system empty on
arrival is a regeneration point.
Therefore, under Assumption \ref{assumption:aperiodic},
$\{(\beta_n,D_n)\}_{n=0,1,\ldots}$ is mixing and ergodic
(cf.\ Remark \ref{remark:mixing-FCFS}).
Theorem \ref{theorem:time-average-GIG1} is thus applicable to the
non-preemptive single-server queues discussed below.
\end{remark}

Under the non-preemptive LCFS service disciplines, the waiting time 
$W_n$ ($n = 1,2,\ldots$) of the $n$th informative packet is 
independent of its service time $H_n^{\dag}$. Note that
the system delay $D_n$ of this packet is given by
\begin{equation}
D_n = W_n + H_n^{\dag}.
\label{eq:D_n-NP-LCFS}
\end{equation}
Furthermore, on the service initiation of an informative packet, 
there exists no waiting packet which can be informative:
no packets are waiting in the discarding case,
while all waiting packets are non-informative in the non-discarding
case. We thus always have $W_n < G_{n+1}^{\dag}$.
Furthermore, if $W_n + H_n^{\dag} > G_{n+1}^{\dag}$, i.e., at least
one packet arrives in the service time of the $n$th informative packet,
the service of the $(n+1)$st informative packet starts immediately after
the service completion. If $W_n + H_n^{\dag} \leq G_{n+1}^{\dag}$, on the
other hand, it takes $\hat{G}_{n+1} + W_{n+1}$ before the service
initiation of the $(n+1)$st informative packet, where 
$\hat{G}_{n+1} := G_{n+1}^{\dag}-W_n-H_n^{\dag}$.

With these observations, we obtain (cf.\
(\ref{eq:alternative-peak-case}))
\begin{align}
A_{\peak,n+1} 
&=
\begin{cases}
W_n + H_n^{\dag} + H_{n+1}^{\dag}, 
& 
W_n + H_n^{\dag} > G_{n+1}^{\dag},
\\[1ex]
\parbox{9em}{
$W_n + H_n^{\dag} + \hat{G}_{n+1}$ 
\\
$\hphantom{W_n} {}+ W_{n+1} + H_{n+1}^{\dag}$,
}
& W_n + H_n^{\dag} \leq G_{n+1}^{\dag}
\end{cases}
\label{eq:NP-LCFS-A_peak-1}
\\
&=
\begin{cases}
W_n + H_n^{\dag} + H_{n+1}^{\dag}, 
&
W_n + H_n^{\dag} > G_{n+1}^{\dag},
\\
G_{n+1}^{\dag} + W_{n+1} + H_{n+1}^{\dag},
&
W_n + H_n^{\dag} \leq G_{n+1}^{\dag}.
\end{cases}
\label{eq:NP-LCFS-A_peak-2}
\end{align}
Note that in non-preemptive LCFS queues with discarding, 
\begin{equation}
W_n + H_n^{\dag} \leq G_{n+1}^{\dag} \Rightarrow W_{n+1} = 0,
\label{eq:W(n+1)=0-w-d}
\end{equation}
so that (\ref{eq:NP-LCFS-A_peak-1}) and (\ref{eq:NP-LCFS-A_peak-2})
are simplified.

We can characterize the distributions of the stationary system delay
$D$ and peak AoI $A_{\peak}$ based on (\ref{eq:D_n-NP-LCFS}), 
(\ref{eq:NP-LCFS-A_peak-1}), and (\ref{eq:NP-LCFS-A_peak-2}),
so that the LST $a^*(s)$ of the AoI is obtained from Theorem 
\ref{theorem:time-average-GIG1} (see Appendix \ref{appendix:NP-LCFS}
for the details).

\begin{theorem}\label{theorem:NP-LCFS}
\begin{itemize}
\item[(i)] In the non-preemptive LCFS M/GI/1 queue with discarding, 
\begin{align}
a^*(s) 
& =
\left(
h^*(\lambda) + \rho \tilde{h}^*(s+\lambda)\right) h^*(s)
\cdot
\frac{\rho \tilde{h}^*(s) + h^*(s+\lambda)\ds\frac{\lambda}{s+\lambda}}
{\rho + h^*(\lambda)}.
\label{eq:a^*-NP-LCFS-MG1-with-d}
\end{align}
\item[(ii)] In the non-preemptive LCFS GI/M/1 queue with discarding,
\begin{align}
a^*(s) 
&=
\biggl[
\tilde{g}^*(s) + \rho \cdot \frac{\mu}{s+\mu} 
\biggl( 
g^*(s+\mu) 
- g^*(\mu) \cdot 
\frac{1-\mu(-g^{(1)}(s+\mu))}{1-\mu(-g^{(1)}(\mu))}
\biggr) \biggr] \frac{\mu}{s+\mu}.
\label{eq:a^*-NP-LCFS-GM1-with-d}
\end{align}
\item[(iii)] In the non-preemptive LCFS M/GI/1 queue without
discarding which satisfies $\rho < 1$,
\begin{align}
a^*(s) 
&=
\frac{\lambda}{s+\lambda} \cdot h^*(s)
\biggl[
\rho \tilde{h}^*(s) 
+
\frac{(1-\rho) (s+\lambda) (1-h^*(s)+h^*(s+\lambda))}
{s + \lambda h^*(s+\lambda)}
\biggr].
\label{eq:a^*-NP-LCFS-MG1-wo-d}
\end{align}
\item[(iv)] In the non-preemptive LCFS GI/M/1 queue without
discarding which satisfies $\rho < 1$,
\begin{equation}
a^*(s) 
=
\left[
\tilde{g}^*(s) + \rho (g^*(s+\mu-\mu\gamma)-\gamma)\frac{\mu}{s+\mu}
\right]
\frac{\mu}{s+\mu},
\label{eq:a^*-NP-LCFS-GM1-wo-d}
\end{equation}
where $\gamma$ denotes the unique solution of (\ref{eq:gamma-def}). 
\end{itemize}
\end{theorem}

Taking the derivatives of $a^*(s)$, we can obtain the moments of AoI.
We provide only the mean AoI below, because formulas for higher
moments are messy.

\begin{corollary}\label{corollary:NP-LCFS}
\begin{itemize}
\item[(i)] In the non-preemptive LCFS M/GI/1 queue with discarding,
\begin{align}
\E[A] 
& =
\frac{1}{\rho+h^*(\lambda)}
\left(
\frac{\lambda \E[H^2]}{2} + \frac{h^*(\lambda)}{\lambda} + (- h^{(1)}(\lambda))
\right)
+ \frac{1-h^*(\lambda)}{\lambda} - (- h^{(1)}(\lambda)) + \E[H].
\label{eq:corollary:EA-NP-LCFS-MG1-with-d}
\end{align}
\item[(ii)] In the non-preemptive LCFS GI/M/1 queue with discarding,
\begin{align}
\E[A] 
&=
\E[H] + \frac{\E[G^2]}{2\E[G]} 
+ 
\rho \biggl(
(-g^{(1)}(\mu)) 
+
\frac{\mu g^*(\mu) g^{(2)}(\mu)}{1-\mu(-g^{(1)}(\mu))} \biggr).
\label{eq:EA-NP-LCFS-GM1-with-d}
\end{align}
\item[(iii)] In the non-preemptive LCFS M/GI/1 queue without
discarding which satisfies $\rho < 1$,
\begin{equation}
\E[A] 
=
\frac{\lambda \E[H^2]}{2} 
+
\left(\frac{(1-\rho)^2}{\rho h^*(\lambda)} + 2\right) \E[H].
\label{eq:EA-NP-LCFS-MG1-w/o-d}
\end{equation}

\item[(iv)] In the non-preemptive LCFS GI/M/1 queue without
discarding which satisfies $\rho < 1$,
\begin{equation}
\E[A] 
= 
\E[H] + \frac{\E[G^2]}{2\E[G]} + \rho (-g^{(1)}(\mu-\mu\gamma)).
\label{eq:EA-NP-LCFS-GM1-w/o-d}
\end{equation}
\end{itemize}
\end{corollary}

As mentioned in Section \ref{sec:intro}, we can evaluate the logging
overhead by comparing the non-preemptive LCFS queues with and without
discarding. In Figure \ref{fig:LCFS-NP}, the mean AoI $\E[A]$ is
plotted as a function of $\rho$ for the non-preemptive LCFS M/M/1 and
D/M/1 queues with and without discarding, where $\E[H]=1$. 
We observe that the logging overhead is relatively large in the 
M/M/1 queue, while in the D/M/1 queue, the minimum values of $\E[A]$
have little difference between the cases with and without discarding
non-informative packets. Therefore, for the deterministic arrival
case, one can collect all sampled data with a small effect on the 
mean AoI $\E[A]$, using the non-preemptive LCFS service discipline
without discarding.

\begin{figure}[tbp]
\centering
\includegraphics[scale=1.0]{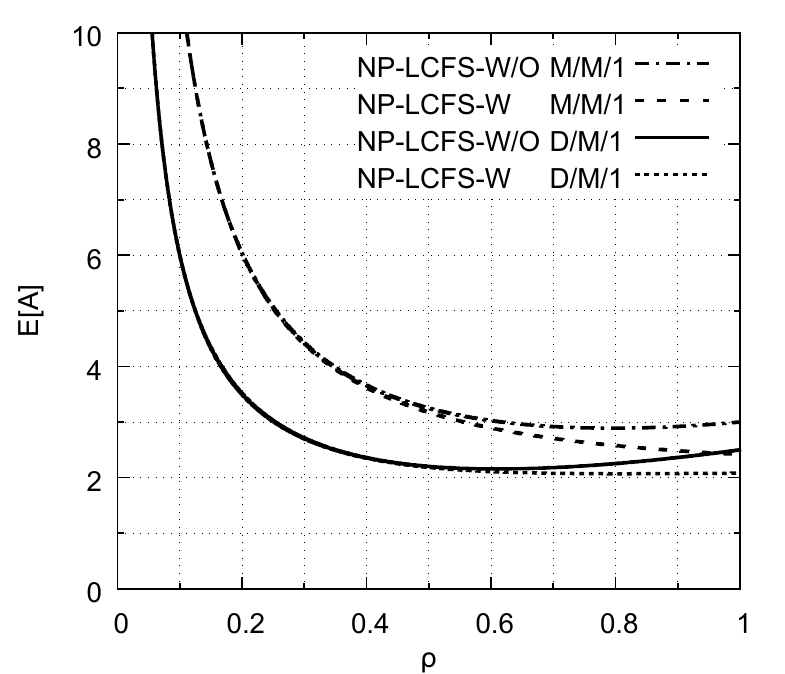}
\caption{$\E[A]$ in the LCFS-NP M/M/1 and M/D/1 queues with and
without discarding.}
\label{fig:LCFS-NP}
\end{figure}

\subsection{Mean AoI comparison among four service disciplines}
\label{ssec:comparison}

In this subsection, we compare the mean AoI in the stationary M/GI/1
and GI/M/1 queues with four service disciplines.
Let $\E[A_{\rm FCFS}]$, $\E[A_{\rm LCFS}^{\rm P}]$,
$\E[A_{\rm LCFS}^{\rm NP-W}]$, and $\E[A_{\rm LCFS}^{\rm NP-W/O}]$
denote the mean AoI in the FCFS queue, preemptive LCFS queue,
non-preemptive LCFS queue with discarding, and
non-preemptive LCFS queue without discarding, respectively.

For the M/GI/1 and GI/M/1 queues, we obtain the following relations,
whose proofs are given in Appendix \ref{appendix:M/G/1-G/M/1-order}.

\begin{theorem}\label{theorem:M/G/1-G/M/1-order}
\begin{itemize}
\item[(i)]
In the stationary M/GI/1 queue, we have
\[
\E[A_{\rm LCFS}^{\rm NP-W}]
\leq
\E[A_{\rm LCFS}^{\rm NP-W/O}]
\leq
\E[A_{\rm FCFS}].
\]

\item[(ii)]
In the stationary GI/M/1 queue, we have
\[
\E[A_{\rm LCFS}^{\rm P}]
\leq
\E[A_{\rm LCFS}^{\rm NP-W}]
\leq
\E[A_{\rm LCFS}^{\rm NP-W/O}]
\leq
\E[A_{\rm FCFS}].
\]
\end{itemize}
\end{theorem}

Theorem \ref{theorem:M/G/1-G/M/1-order} (ii) shows that 
for exponential service times, $\E[A_{\rm LCFS}^{\rm
P}]$ is the smallest among four service disciplines,
which is almost obvious from the memoryless property of the
exponential distribution.
Readers are referred to \cite{Bedewy16} for a detailed discussion on
the optimality of the preemptive LCFS discipline in queues with
exponential service times.

\begin{figure}[t]
\centering
\includegraphics[scale=1.0]{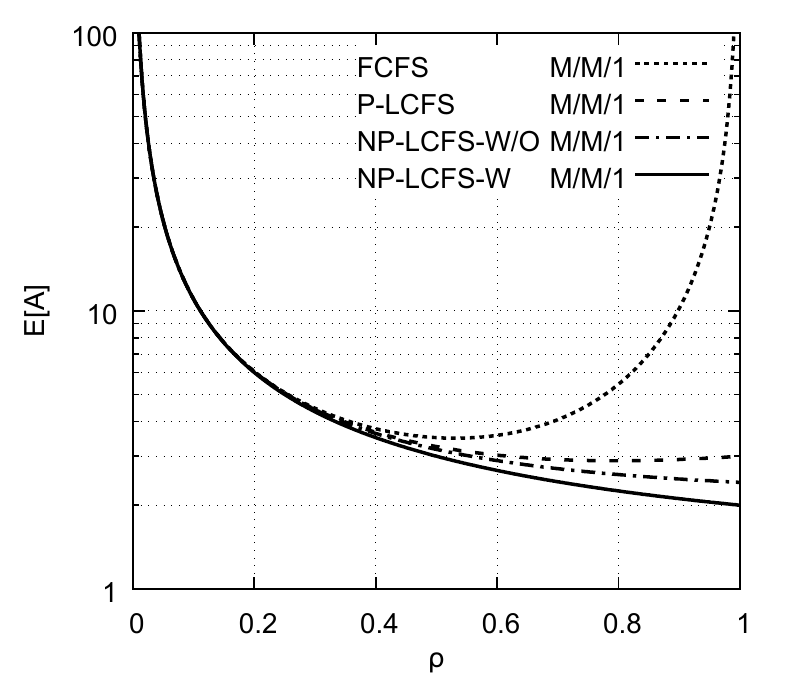}
\caption{Comparison of the mean AoI in M/M/1 queues.}
\label{fig:m-m-1}
\mbox{}
\\
\includegraphics[scale=1.0]{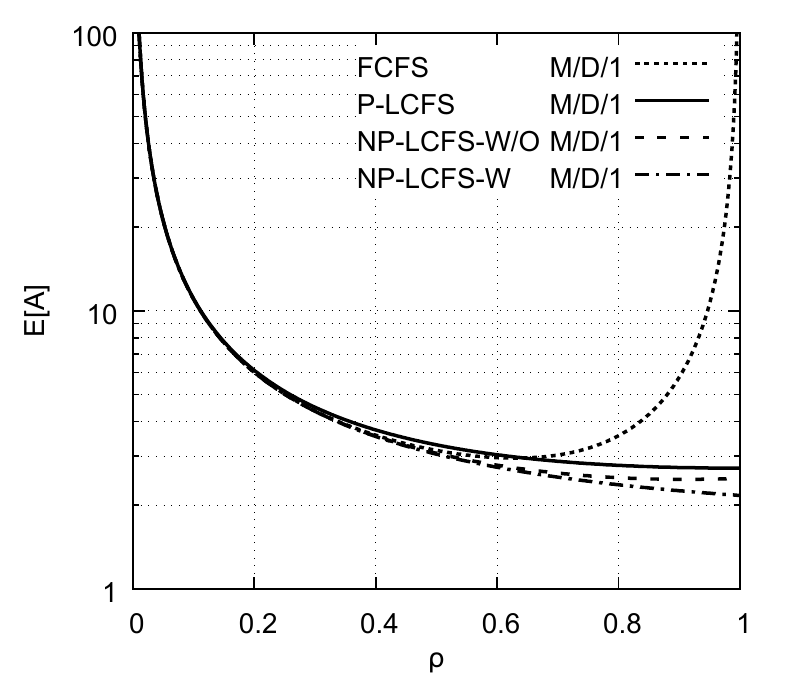}
\caption{Comparison of the mean AoI in M/D/1 queues.}
\label{fig:m-d-1}
\end{figure}

Figures \ref{fig:m-m-1} and \ref{fig:m-d-1} show the mean AoI
in the M/M/1 and M/D/1 queues with $\E[H]=1$ as a function of the
traffic intensity $\rho$.  We observe that when the traffic intensity
$\rho$ is large, $\E[A_{\rm FCFS}]$ increases with $\rho$ because it
diverges as $\rho \to 1$ (cf.\ (\ref{eq:mean-age-MG1})). 
On the other hand, $\E[A_{\rm LCFS}^{\rm NP-W/O}]$ takes a moderate
value for large $\rho$, even though the system delay of
non-informative packets diverges as $\rho \to 1$. 

As discussed in Section \ref{ssec:LCFS-P}, the mean AoI in the
preemptive M/GI/1 queue is influenced strongly by the service time
distribution. Particularly, in the M/D/1 queue, using the results in
Appendix \ref{appendix:summary}, we can analytically show that
$\E[A_{\rm LCFS}^{\rm P}] \geq \E[A_{\rm LCFS}^{\rm NP-W/O}]$
for all $\rho \in (0,1)$ and 
\[
\E[A_{\rm LCFS}^{\rm P}] \geq \E[A_{\rm FCFS}],
\quad
\rho \in (0,\hat{\rho}^*], 
\]
where $\hat{\rho}^* \approx 0.643798$ denotes the unique positive
solution of $2(1-\rho)\exp(\rho)+\rho-2 = 0$.  

We close this section with the following theorem whose
proof is given in Appendix \ref{appendix:theorem:MG1-FCFS<P-LCFS}

\begin{theorem}\label{theorem:MG1-FCFS<P-LCFS}
In the stationary M/GI/1 queue, we have
\begin{align}
\lefteqn{\rho \in (0,2-\sqrt{2})\ \mbox{and}\ (\Cv[H])^2 \leq v(\rho)}
&
\nonumber
\\
& \qquad\qquad\qquad\qquad {}
\Rightarrow\ 
\E[A_{\rm LCFS}^{\rm P}] \geq \E[A_{\rm FCFS}],
\label{eq:theorem:MG1-FCFS<P-LCFS}
\end{align}
where $v(\rho)$ ($\rho \in (0,2-\sqrt{2})$) is a decreasing function of $\rho$,
which is given by
\[
v(\rho)
= 
\frac{1}{\rho^2} \left(\sqrt{1+\rho^2 \{(2-\rho)^2 - 2\}}-1\right) 
< 
v(0+)=1.
\]
\end{theorem}

Figure \ref{figure:sufficient-region} depicts the region 
(\ref{eq:theorem:MG1-FCFS<P-LCFS}) of
($\rho$, $\Cv[H]$) for $\E[A_{\rm LCFS}^{\rm P}] \geq
\E[A_{\rm FCFS}]$.

\begin{figure}[t]
\centering
\includegraphics[scale=1.0]{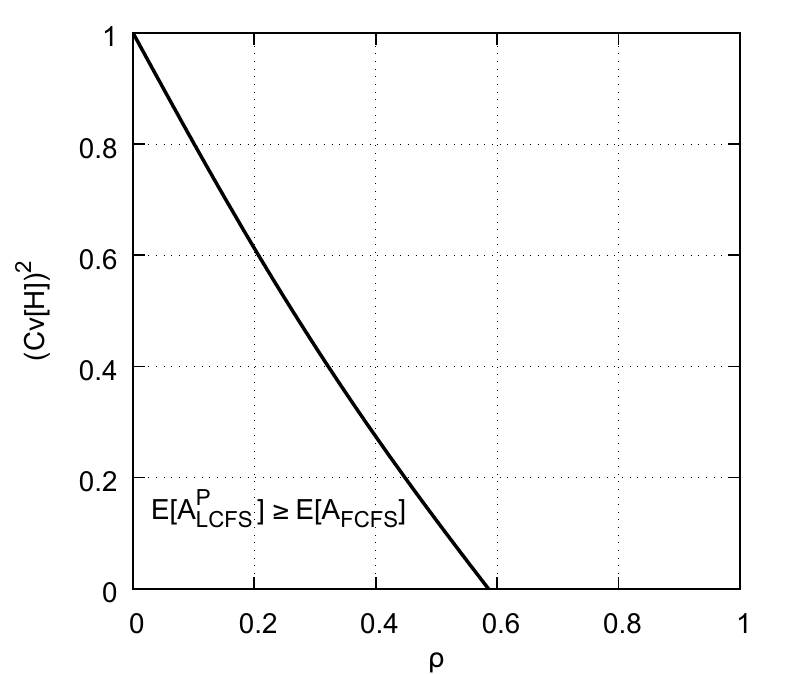}
\caption{Sufficient condition (\ref{eq:theorem:MG1-FCFS<P-LCFS}) 
in terms of $\rho$ and $(\Cv[H])^2$.}
\label{figure:sufficient-region}
\end{figure}

\section{Conclusion} \label{sec:conclusion}

We presented a general formula for the stationary distribution of the
AoI, which holds for a wide class of information update systems.  The
formula shows that the stationary distribution of the AoI is given in
terms of the stationary distributions of the system delay and the peak
AoI. Therefore it provides a unified, efficient approach to the
analysis of the AoI because the system delay and the peak AoI can be
analyzed by standard techniques in queueing theory.  To demonstrate
this fact, we analyzed the stationary distributions of the AoI in
single-server queues under four different service disciplines:
first-come first-served (FCFS), preemptive last-come first-served
(LCFS), and non-preemptive LCFS with and without discarding. Moreover,
we showed comparison results for the mean AoI in the M/GI/1 and GI/M/1
queues under these service disciplines. The results in this paper will
form a basis for the analysis of the AoI in developing sophisticated
information update systems.

\appendix

\section{Summary of Results for Special Cases}\label{appendix:summary}

In this appendix, we summarize some simplified formulas for three
special cases: The M/M/1, M/D/1, and D/M/1 queues.  Note that formulas
for the M/M/1 queue can be obtained from those in either the M/GI/1 or
GI/M/1 queues, and formulas for the M/D/1 (resp.\ D/M/1) queue can be
obtained from those for the M/GI/1 (resp.\ GI/M/1) queue.

\subsection{FCFS Queues}\label{appendix:FCFS}

\subsubsection*{The FCFS M/M/1 queue}

Consider the stationary FCFS M/M/1 queue.
From (\ref{eq:LST-M/GI/1}), we have
\begin{align*}
a^*(s) 
&=
\frac{(1-\rho)\mu}{s+(1-\rho)\mu} 
-
\frac{(1-\rho) \mu s (s+\lambda+\mu)}{(s+\mu)^2 (s+\lambda)},
\end{align*}
and therefore,
\begin{align*}
A(x)
&=
1 - e^{-(1-\rho) \mu x} 
+
\left( \frac{1}{1-\rho} + \rho \mu x \right) e^{-\mu x} 
- 
\frac{1}{1-\rho} \cdot e^{-\lambda x}.
\end{align*}
Because $\E[H^2]=2(\E[H])^2$ holds in this model, (\ref{eq:mean-age-MG1}) and
(\ref{eq:2nd-age-MG1}) are also reduced to
\begin{align}
\E[A] 
&= 
\left( \frac{1}{1-\rho} + \frac{1}{\rho} - \rho \right) \E[H],
\label{eq:mean-age-MM1}
\\
\E[A^2] 
&=
2 \left(
\frac{1}{(1-\rho)^2} -2\rho + \frac{1}{\rho} + \frac{1}{\rho^2}
\right) (\E[H])^2.
\nonumber
\end{align}
Note that (\ref{eq:mean-age-MM1}) is identical to Eq.\ (17) 
in \cite{Kaul12-1}.

\subsubsection*{The FCFS M/D/1 queue}

Consider the stationary FCFS M/D/1 queue.
From  (\ref{eq:mean-age-MG1}) and (\ref{eq:2nd-age-MG1}), 
we have
\begin{align}
\E[A] 
&=
\left(
\frac{1}{2(1-\rho)} 
+
\frac{1}{2}
+ 
\frac{(1-\rho) \exp(\rho)}{\rho}
\right) \E[H],
\label{eq:mean-age-MD1}
\\
\E[A^2] 
&=
\left(
\frac{1}{2(1-\rho)^2} + \frac{1}{3(1-\rho)} + \frac{1}{6}
+ \frac{2(1-\rho) \exp(2\rho)}{\rho^2}
\right)(\E[H])^2.
\notag
\end{align}
Although the FCFS M/D/1 queue is considered
in \cite{Kaul12-1}, no explicit formula for $\E[A]$ is provided there.
With the explicit formula (\ref{eq:mean-age-MD1}), one can deduce
that $E[A]$ is minimized when $\rho \approx 0.6291$.

In addition, with (\ref{eq:mean-age-MM1}) and (\ref{eq:mean-age-MD1}),
it is easy to verify that $\E[A]$ in the FCFS M/D/1 queue
is strictly smaller than that in the FCFS M/M/1 queue with the same
$\E[H]$ and $\rho$ ($0 < \rho < 1$).

\subsubsection*{The FCFS D/M/1 queue}

Consider the stationary FCFS D/M/1 queue.
From (\ref{eq:mean-age-GM1}) and (\ref{eq:2nd-age-GM1}), 
we have
\begin{align}
\E[A]
&=
\left( \frac{1}{2\rho} + \frac{1}{1-\gamma} 
\right) \E[H],
\label{eq:mean-age-DM1}
\\
\E[A^2]
&=
\left( 
2 \left(\frac{1}{1-\gamma}\right)^2 
+
\frac{1}{(1-\gamma)\rho} 
+
\frac{1}{3 \rho^2}
\right) (\E[H])^2.
\nonumber
\end{align}
We note that (\ref{eq:mean-age-DM1}) is identical to Eq.\ (25) in
\cite{Kaul12-1}.

\subsection{Preemptive LCFS queues}\label{appendix:PR-LCFS}

\subsubsection*{The preemptive LCFS M/M/1 queue}

Consider the preemptive LCFS M/M/1 queue.
In this case, (\ref{eq:theorem-PR-LCFS-AoI}) is reduced to 
\[
a^*(s) = \frac{\lambda}{s+\lambda} \cdot \frac{\mu}{s+\mu},
\]
so that the AoI is given by the sum of an interarrival time and a 
service time, which are independent. We thus have
if $\lambda \neq \mu$, 
\[
\Pr(A \leq x) 
= 
1 - \frac{\mu}{\mu-\lambda} \cdot e^{-\lambda x}
+ \frac{\lambda}{\mu-\lambda} \cdot e^{-\mu x}, 
\quad x \geq 0,
\]
and if $\lambda=\mu$, the AoI $A$ follows an Erlang distribution of
the second order.
\[
\Pr(A \leq x) 
= 
1 - e^{\lambda x} - \lambda x e^{-\lambda x}, 
\quad x \geq 0.
\]
Furthermore,
\begin{align}
\E[A] &= \left(1+\frac{1}{\rho}\right) \E[H],
\label{eq:EA-preemptive-MM1}
\\
\E[A^2] &= 2\left(1+\frac{1}{\rho}+\frac{1}{\rho^2}\right) (\E[H])^2.
\nonumber
\end{align}
Note that (\ref{eq:EA-preemptive-MM1}) is identical to Eq.\ (48) in 
\cite{Kaul12-2}.

\subsubsection*{The preemptive LCFS M/D/1 queue}
\label{appendix:LCFS-P-MD1}

Consider the preemptive LCFS M/D/1 queue.
From Corollary \ref{corollary:PR-LCFS-summary}, 
we have
\begin{align}
\E[A] &= \frac{\exp(\rho)}{\rho} \cdot \E[H],
\label{eq:meanAoI-P-LCFS-MD1}
\\
\E[A^2] &= 
2 \left(\frac{\exp(\rho)}{\rho} - 1\right) 
\frac{\exp(\rho)}{\rho} (\E[H])^2. 
\nonumber
\end{align}

\subsubsection*{The preemptive LCFS D/M/1 queue}

Consider the preemptive LCFS D/M/1 queue.  From
(\ref{eq:mean-AoI-GM1-LCFS-P}), we have
\begin{align*}
\E[A] &= \frac{\E[G]}{2} + \E[H],
\\
\E[A^2] &= \frac{(\E[G])^2}{3} + \E[H]\E[G] + 2(\E[H])^2.
\end{align*}

\subsection{Non-preemptive LCFS queues with discarding}
\label{subsection:appendix-NP-LCFS-C}

\subsubsection*{The non-preemptive LCFS M/M/1 queue with discarding}

Consider the non-preemptive LCFS M/M/1 queue with discarding.
From Theorem \ref{theorem:NP-LCFS} and Corollary
\ref{corollary:NP-LCFS}, we have
\begin{align}
a^*(s)
&= 
\frac{1+\rho}{1+\rho+\rho^2}
\left( \frac{\lambda}{s+\lambda} + \rho \cdot \frac{\mu}{s+\mu}
- \frac{\rho}{1+\rho} \cdot \frac{\lambda+\mu}{s+\lambda+\mu}
\right)
\left(\frac{1}{1+\rho}+\frac{\rho}{1+\rho}\cdot \frac{\lambda+\mu}{s+\lambda+\mu}
\right) \frac{\mu}{s+\mu},
\nonumber
\\
\E[A]
&=
\frac{\E[H]}{1+\rho+\rho^2}
\left[
2\rho^2 + 3\rho + \frac{1}{\rho}+\frac{3}{1+\rho} 
- \left(\frac{1}{1+\rho}\right)^2
\right].
\label{EA-NP-MM1-w}
\end{align}
We can verify that (\ref{EA-NP-MM1-w}) is identical to Eq.\ (65) in 
\cite{Costa16}.

\subsubsection*{The non-preemptive LCFS M/D/1 queue with discarding}

Consider the non-preemptive LCFS M/D/1 queue with
discarding. From Corollary \ref{corollary:NP-LCFS}, 
we have
\begin{align*}
\E[A] 
&= 
\left[
\frac{1}{1+\rho \exp(\rho)}\left(\frac{1}{2}+\frac{1}{\rho}\right)
+
\frac{\exp(\rho)-(1+\rho)}{\rho \exp(\rho)} + \frac{3}{2}
\right]\E[H].
\end{align*}

\subsubsection*{The non-preemptive LCFS D/M/1 queue with discarding}

Consider the stationary non-preemptive LCFS D/M/1 queue with
discarding. From Corollary \ref{corollary:NP-LCFS},
we have
\begin{align*}
\E[A] &= \frac{\E[G]}{2} + \frac{\E[H]}{1 - \rho^{-1}\exp(-\rho^{-1})}.
\end{align*}

\subsection{Non-preemptive LCFS queues without discarding}
\label{subsection:appendix-NP-LCFS-D}

\subsubsection*{The non-preemptive LCFS M/M/1 queue without discarding}

Consider the stationary non-preemptive LCFS M/M/1 queue without
discarding. From Theorem \ref{theorem:NP-LCFS} and Corollary 
\ref{corollary:NP-LCFS}, we have
\begin{align*}
a^*(s) 
&=
\left[
1 - \rho + \rho(2-\rho) \frac{\mu}{s+\mu} 
- \rho(1-\rho) \left(\frac{\mu}{s+\mu}\right)^2
\right]
\frac{\lambda}{s+\lambda} \cdot \frac{\mu}{s+\mu},
\\
\E[A]
&=
\left(\rho^2 + 1+\frac{1}{\rho}\right) \E[H],
\\
\E[A^2]
&=2\left( 3\rho^2 + 1 +\frac{1}{\rho}+\frac{1}{\rho^2} \right)
(\E[H])^2.
\end{align*}

\subsubsection*{The non-preemptive LCFS M/D/1 queue without discarding}

Consider the stationary non-preemptive LCFS M/D/1 queue without
discarding.  From Corollary \ref{corollary:NP-LCFS},
we have
\begin{equation}
\E[A] = \left(
\frac{\rho}{2} + 2 + \frac{(1-\rho)^2}{\rho} \cdot \exp(\rho)\right) \E[H].
\label{eq:meanAoI-NP-LCFS-WO-MD1}
\end{equation}

\subsubsection*{The non-preemptive LCFS D/M/1 queue without discarding}

Consider the stationary non-preemptive LCFS D/M/1 queue without
discarding. In this case, we have $-g^{(1)}(\mu-\mu\gamma) = \gamma\E[G]$. 
Corollary \ref{corollary:NP-LCFS} then yields
\[
\E[A] = \frac{\E[G]}{2} + (1+\gamma)\E[H].
\]

\section{Proof of Lemma \ref{lemma:cov-D-G}}
\label{appendix:cov-D-G}

We first prove 
\begin{align}
\Cov[G_n^{\dag},D_n]
&= 
\int_0^{\infty} 
\E\left[(y - G)(G-\E[G]) \mid G \leq y\right]
G(y) \dd D(y).
\label{eq:cov-D-G}
\end{align}
It follows from Lindley's recursion that
\begin{equation}
D_n = \one_{\{G_n^{\dag} \leq D_{n-1}\}}(D_{n-1} - G_n^{\dag}) + H_n^{\dag},
\quad 
n = 1,2,\ldots.
\label{recursion-D_n}
\end{equation}
Since both $D_{n-1}$ and $H_n^{\dag}$ are independent of $G_n^{\dag}$, it
follows from (\ref{recursion-D_n}) and the stationarity of the system that
\begin{align}
\E[G_n^{\dag} D_n] - \E[G]\E[H]
&=
\E[G_n^{\dag} (D_n-\E[H])]
\nonumber 
\\
&=
\int_0^{\infty} 
\E\left[\one_{\{G \leq y\}} (y - G)G\right] \dd D(y)
\nonumber
\\
&=
\int_0^{\infty} 
G(y) \E\left[ (y - G)G \,|\, G \leq y \right] \dd D(y).
\label{hm-eqn-05}
\end{align}
Similarly, we have
\begin{align}
\E[G_n^{\dag}] \E[D_n] - \E[G]\E[H]
&=\int_0^{\infty} 
G(y) \E\left[y - G \,|\, G \leq y \right] \E[G] \dd D(y).
\label{hm-eqn-06}
\end{align}
Substituting (\ref{hm-eqn-05}) and (\ref{hm-eqn-06}) into
$\Cov[G_n^{\dag},D_n] = \E[G_n^{\dag} D_n] - \E[G]\E[D]$ yields (\ref{eq:cov-D-G}).

Using (\ref{eq:cov-D-G}), we obtain the lower bound of
$\Cov[G_n^{\dag},D_n]$
as follows:
\begin{align*}
\Cov[G_n^{\dag},D_n]
&=
\int_0^{\infty} 
\E\left[(y - G)(G-\E[G])\one_{\{G < \E[G]\}}\one_{\{G \leq y\}}\right]
\dd D(y)
\\
&\quad
{}+
\int_0^{\infty} 
\E\left[(y - G)(G-\E[G])\one_{\{\E[G] \leq G \leq y\}}\right]
\dd D(y)
\\
&\geq
\int_0^{\infty} 
\E\left[(y - G)(G-\E[G])\one_{\{G < \E[G]\}}\one_{\{G \leq y\}}\right]
\dd D(y)
\\
&=
\int_0^{\infty} 
\E\bigl[ ( -\E[G]y + G \cdot (y-G+\E[G])) 
\\
& \qquad\qquad\qquad\qquad\qquad\qquad {} 
\times
\one_{\{G < \E[G]\}}\one_{\{G \leq y\}}\bigr]
\dd D(y)
\\
&\geq
- \E[G]\int_0^{\infty} 
\E\left[\one_{\{G < \E[G]\}}\one_{\{G \leq y\}}\right]
y \dd D(y)
\\
&\geq
- \E[G]\int_0^{\infty} \E[\one_{\{G < \E[G]\}}] y \dd D(y)
\\
&=
- \E[G]\Pr(G < \E[G] )\E[D].
\end{align*}

Finally, we prove $\Cov[G_n^{\dag},D_n] \leq 0$. Consider an arbitrarily fixed
$y$ such that $G(y)=\Pr(G \leq y) > 0$. If $\Pr(G > y) =0$, then $\E[G]=\E[G
\mid G \leq y]$, and otherwise
\begin{align*}
\E[G] 
&
= \Pr(G \leq y) \E[G \mid G \leq y] + \Pr(G > y) \E[G \mid G > y]
\\
& \geq
\Pr(G \leq y) \E[G \mid G \leq y] + \Pr(G > y) \E[G \mid G \leq y]
\\
& =
\E[G \mid G \leq y],
\end{align*}
because $\E[G \mid G \leq y] \leq y \leq \E[G \mid G > y]$ if $\Pr(G >
y) > 0$.  Note also that $(y-x)x$ is concave with respect to $x$.
Using these facts and Jensen's inequality, we have
\begin{align}
\E\left[(y - G)G \mid G \leq y\right]
&\leq
\left(y - \E[G \mid G \leq y]\right) \E[G \mid G \leq y]
\nonumber
\\
&\leq
\E[y - G \mid G \leq y] \E[G],
\label{hm-eqn-07}
\end{align}
for all $y \geq 0$ such that $G(y) > 0$. 
Combining (\ref{hm-eqn-05}), (\ref{hm-eqn-06}) and (\ref{hm-eqn-07}),
we obtain $\E[G_n^{\dag} D_n] \leq \E[G_n^{\dag}]\E[D_n]$ and thus
$\Cov[G_n^{\dag},D_n] \leq 0$. 

\section{Derivation of (\ref{eq:PR-LCFS-GG1-a^*(s)})}
\label{appendix:PR-LCFS-GG1-a^*(s)}

From (\ref{eq:PR-LCFS-GG1-(n+1)st-peak}), we obtain
\begin{align}
a_{\peak}^*(s) 
&= 
g_{>H}^*(s) 
\sum_{m=1}^{\infty} (1-\zeta)\zeta^{m-1} \cdot (g_{<H}^*(s))^{m-1} h_{<G}^*(s)
\nonumber
\\
&=
\frac{(1-\zeta) g_{>H}^*(s)}{1-\zeta g_{<H}^*(s)} \cdot h_{<G}^*(s).
\label{eq:PR-LCFS-GG1-peak}
\end{align}
With straightforward calculations using (\ref{eq:lmd^dag-Ap-D}), 
(\ref{eq:D_n-PR-LCFS-GG1}), and (\ref{eq:PR-LCFS-GG1-peak}), 
we obtain 
\begin{equation}
\lambda^{\dag} = \frac{1-\zeta}{\E[G]}.
\label{eq:Pr-LCFS-GG1-throughput}
\end{equation}
Therefore, (\ref{eq:PR-LCFS-GG1-a^*(s)}) follows from 
Theorem \ref{theorem:time-average-GIG1},
(\ref{eq:D_n-PR-LCFS-GG1}), (\ref{eq:PR-LCFS-GG1-peak}),
and (\ref{eq:Pr-LCFS-GG1-throughput}). 
Furthermore, taking the derivative of $a^*(s)$, 
we obtain $\E[A]$.

\section{Proof of Theorem \ref{theorem:NP-LCFS}}
\label{appendix:NP-LCFS}

\subsection{Derivation of (\ref{eq:a^*-NP-LCFS-MG1-with-d})}

It is readily seen that the waiting time $W_n$ ($n=2,3,\ldots$) of the
$n$th informative packet is given by
\begin{equation}
W_n = \begin{cases}
X(H_{n-1}^{\dag}), & \mbox{with probability $1-h^*(\lambda)$},
\\
0, & \mbox{with probability $h^*(\lambda)$},
\label{eq:W_n-LCFS-NP-MG1-w-d}
\end{cases}
\end{equation}
where $X(H_{n-1}^{\dag})$ denotes the remaining service time 
seen by the last packet arrived in $H_{n-1}^{\dag}$.
We can verify that the LST $\psi^*(s)$ of $X(H_{n-1}^{\dag})$ is given by
\begin{align}
\psi^*(s)
&= 
\frac{1}{1-h^*(\lambda)}
\int_{x=0}^{\infty} \dd H(x)
\sum_{k=1}^{\infty}
\int_{y=0}^x 
\frac{e^{-\lambda y}(\lambda y)^{k-1}\lambda}{(k-1)!}
\cdot e^{-\lambda (x-y)} e^{-s(x-y)} \dd y 
\nonumber
\\
&=
\frac{\lambda}{s+\lambda} \cdot \frac{1-h^*(s+\lambda)}{1-h^*(\lambda)},
\label{eq:lemma:psi(s)}
\end{align}
so that we have from (\ref{eq:W_n-LCFS-NP-MG1-w-d}),
\[
w^*(s) 
= h^*(\lambda) + \frac{\lambda}{s+\lambda}( 1-h^*(s+\lambda) ).
\]
Since service times are i.i.d.\ and services are non-preemptive, 
$W_n$ and $H_n^{\dag}$ are independent. Therefore, we obtain from
(\ref{eq:D_n-NP-LCFS}),
\begin{equation}
d^*(s) = w^*(s)h^*(s) 
=
\left(h^*(\lambda) + \rho \tilde{h}^*(s+\lambda)\right) h^*(s).
\label{eq:NP-LCFS-with-d^*(s)}
\end{equation}

Using (\ref{eq:NP-LCFS-A_peak-1}) and (\ref{eq:W(n+1)=0-w-d}), we can
verify that
\begin{align}
a_{\peak}^*(s)
&=
(1 - h^*(\lambda))
\cdot
w^*(s) \cdot \frac{h^*(s)-h^*(s+\lambda)}{1-h^*(\lambda)} \cdot h^*(s)
\nonumber
\\
&\quad
{}+
h^*(\lambda)
\cdot
w^*(s) \cdot \frac{h^*(s+\lambda)}{h^*(\lambda)} 
\cdot \frac{\lambda}{s+\lambda} \cdot h^*(s)
\nonumber
\\
&=
d^*(s) \left[
h^*(s) - h^*(s+\lambda) \left(1-\frac{\lambda}{s+\lambda}\right)
\right].
\label{eq:a_peak-NP-LCFS-with-d}
\end{align}
We then obtain (\ref{eq:a^*-NP-LCFS-MG1-with-d})
with straight forward calculations based on 
(\ref{eq:time-average-GIG1-LST}), (\ref{eq:lmd^dag-Ap-D}),
(\ref{eq:NP-LCFS-with-d^*(s)}), and (\ref{eq:a_peak-NP-LCFS-with-d}).

\subsection{Derivation of (\ref{eq:a^*-NP-LCFS-GM1-with-d})}

Note first that 
\[
\Pr(W_{n+1} > 0 \mid W_n = 0) 
=
\int_0^{\infty} e^{-\mu x} \dd G(x)
= g^*(\mu),
\]
and 
\begin{align*}
\Pr(W_{n+1} > 0 \mid W_n > 0)
&= 
\frac{1}{1-g^*(\mu)}
\int_{x=0}^{\infty} \dd G(x) \int_{y=0}^x \mu e^{-\mu y} \cdot e^{-\mu (x-y)} \dd y 
\\
&=
- \frac{\mu g^{(1)}(\mu)}{1-g^*(\mu)}. 
\end{align*}
Therefore, considering the stationary distribution of 
a discrete-time Markov chain with two states
``no wait'' and ``wait'', we obtain
\begin{equation}
\Pr(W = 0) 
= 
\ds\frac{q}{q + g^*(\mu)},
\quad
\Pr(W > 0) 
=
\ds\frac{g^*(\mu)}
{q + g^*(\mu)},
\label{eq:W-prob-LCFS-NP-w-d}
\end{equation}
where $q$ is given by
\[
q = 1 - \frac{\mu (-g^{(1)}(\mu))}{1-g^*(\mu)}.
\]

We then have from (\ref{eq:D_n-NP-LCFS}), 
\begin{align}
d^*(s) 
&=
\frac{\mu}{s+\mu} \biggl[\Pr(W = 0) 
+
\frac{\Pr(W > 0) }{1-g^*(\mu)}
\int_0^{\infty} e^{-s x} \mu e^{-\mu x} (1-G(x)) \dd x 
\biggr]
\nonumber
\\
&=
\frac{\mu}{s+\mu}
\biggl[
\Pr(W=0) 
+ 
\Pr(W > 0) \cdot \frac{\mu}{s+\mu}
\cdot \frac{1-g^*(s+\mu)}{1-g^*(\mu)}
\biggr].
\label{eq:NP-LCFS-GM1-W-d^*(s)}
\end{align}

Let $a_{\peak,0}^*(s)$ (resp.\ $a_{\peak,+}^*(s)$) denote the
conditional LST of the $(n+1)$st peak AoI $A_{\peak,n+1}$ given
$W_n=0$ (resp.\ $W_n > 0$):
\begin{align}
a_{\peak}^*(s)
&=
\Pr(W=0) a_{\peak,0}^*(s) 
+ 
\Pr(W > 0) a_{\peak,+}^*(s).
\label{eq:a_peak-NP-GM1}
\end{align}
When $W_n=0$, it follows from (\ref{eq:NP-LCFS-A_peak-2}) and
(\ref{eq:W(n+1)=0-w-d}) that
\[
A_{\peak,n+1} = \max(H_n^{\dag}, G_{n+1}^{\dag}) + H_{n+1}^{\dag},
\]
which implies
\begin{align}
a_{\peak,0}^*(s) 
&= 
g^*(s+\mu) \left(\frac{\mu}{s+\mu}\right)^2 
+ 
(g^*(s) - g^*(s+\mu)) \frac{\mu}{s+\mu}.
\label{eq:NP-LCFS-GM1-W-0}
\end{align}

On the other hand, when $W_n > 0$, $W_n$ has the same distribution as
the conditional service time $H_{<G}$ given that $H < G$, where $H$ is
exponentially distributed with parameter $\mu$. From 
(\ref{eq:NP-LCFS-A_peak-2}) and (\ref{eq:W(n+1)=0-w-d}),
the peak AoI is then given by
\begin{align*}
A_{\peak,n+1} 
&\eqdist
\max(H_{<G} + H_n^{\dag}, G) + H_{n+1}^{\dag}
\\
&=
H_{<G} + \max(H_n,G-H_{<G}) + H_{n+1}^{\dag}.
\end{align*}
We define $f^{**}(s,\omega) = \E[e^{-s H_{<G}} e^{-\omega (G -
H_{<G})}]$ as the joint LST of the waiting time $H_{<G}$ and the
remaining interarrival time $G-H_{<G}$, which is given by
\begin{align*}
f^{**}(s,\omega)
&=
\int_0^{\infty} \frac{\dd G(x)}{1-g^*(\mu)} 
\int_0^x e^{-s y} e^{-\omega (x-y)} \mu e^{-\mu y} \dd y
\\
&=
\frac{\mu}{1-g^*(\mu)} \cdot 
\frac{g^*(\omega)-g^*(s+\mu)}{s+\mu-\omega}.
\end{align*}
It can be verified from the above observations that the conditional
LST $a_{\peak,+}^*(s)$ of the peak AoI is given by
\begin{align}
a_{\peak,+}^*(s)
& =
\Bigl[
f^{**}(s,s+\mu) \frac{\mu}{s+\mu} 
+ f^{**}(s,s) - f^{**}(s,s+\mu)\Bigl] \frac{\mu}{s+\mu}
\nonumber
\\
&=
\frac{\mu}{s+\mu}\biggl[
\frac{\mu (-g^{(1)}(s+\mu))}{1-g^*(\mu)} 
\cdot \frac{\mu}{s+\mu}
+
\frac{g^*(s)-g^*(s+\mu)}{1-g^*(\mu)}
-
\frac{\mu (-g^{(1)}(s+\mu))}{1-g^*(\mu)} 
\biggr],
\label{eq:NP-LCFS-GM1-W-+}
\end{align}
where we used
\[
f^{**}(s,s+\mu)
=
\lim_{\omega \to s+\mu} f^{**}(s,\omega)
=
\frac{\mu \cdot (-g^{(1)}(s+\mu))}{1-g^*(\mu)}.
\]
We then obtain (\ref{eq:a^*-NP-LCFS-GM1-with-d})
with some calculations based on 
(\ref{eq:time-average-GIG1-LST}), (\ref{eq:lmd^dag-Ap-D}),
(\ref{eq:NP-LCFS-GM1-W-d^*(s)}), (\ref{eq:a_peak-NP-GM1}),
(\ref{eq:NP-LCFS-GM1-W-0}), and (\ref{eq:NP-LCFS-GM1-W-+}).

\subsection{Derivation of (\ref{eq:a^*-NP-LCFS-MG1-wo-d})}

Because the non-preemptive LCFS service discipline is work-conserving,
the stationary queue length distribution is identical to that in the FCFS M/GI/1 queue. 
Let $L$ denote the number of packets, and let $\tilde{H}$ denote the remaining
service time in steady state provided that a packet is being
served. Note that $\Pr(L \geq 1) = \rho$. Furthermore, it is known
that \cite{Taki92,Wish60}
\begin{align}
\Pi(z,s) 
&:=
\E[z^L e^{-s \tilde{H}} \mid L \geq 1]
\nonumber
\\
&=
\frac{(1-\rho) z (z-1)}{z-h^*(\lambda - \lambda z)}
\cdot
\frac{h^*(\lambda - \lambda z)-h^*(s)}
{\E[H](s-\lambda+\lambda z)}.
\label{eq:Pi(z,s)}
\end{align}

Let $W$ denote the waiting time of informative packets and
$L^{\dag}$ denote the number of waiting (non-informative) packets at
the beginnings of services of the informative packets. 
Arriving packets become informative if (i) they arrive at the empty
system or (ii) they arrive when a packet is being served and no subsequent
packets arrive in the remaining service time.
Therefore, owing to PASTA, we obtain
\begin{align}
w^{**}(z,s) 
&:=
\E[z^{L^{\dag}} e^{-sW}] 
\nonumber
\\
&= \frac{1}{1-\rho + \rho \tilde{h}^*(\lambda)}
\left[ 1-\rho + \rho \cdot \frac{\Pi(z,s+\lambda)}{z} \right].
\label{eq:w**(z,s)}
\end{align}
It then follows from $d^*(s) = w^{**}(1,s) h^*(s)$ and 
(\ref{eq:Pi(z,s)}) that
\begin{equation}
d^*(s)
=
\frac{1-\rho + \rho \tilde{h}^*(s+\lambda) }
{1-\rho + \rho \tilde{h}^*(\lambda)} \cdot h^*(s).
\label{eq:NP-LCFS-MG1-wo-d^*(s)}
\end{equation}

Next, we consider the peak AoI.
We define $w_k^*(s)$ ($k=0,1,\ldots$) as
\[
w_k^*(s) = \Pr( L^{\dag} = k) \E[e^{-sW} \mid L^{\dag}=k].
\]
By definition, $w^{**}(z,s) = \sum_{k=0}^{\infty} w_k^*(s)
z^k$. With straightforward calculations based on (\ref{eq:NP-LCFS-A_peak-1}), 
we can verify that
\begin{align}
a_{\peak}^*(s)
&=
(1 - h^*(\lambda))
\cdot 
w^{**}(1,s) \cdot \frac{h^*(s)-h^*(s+\lambda)}{1-h^*(\lambda)} \cdot h^*(s)
\nonumber
\\
&\quad
{}+
h^*(\lambda)
\cdot
\frac{1}{h^*(\lambda)}
\sum_{k=0}^{\infty}
w_k^*(s) h^*(s+\lambda) 
\left[
\sum_{\ell=1}^k (h^*(s+\lambda))^{\ell-1} (h^*(s)-h^*(s+\lambda)) h^*(s)
\right.
\nonumber
\\
&\hspace{25em} {} \left.
+
(h^*(s+\lambda))^k \cdot \frac{\lambda}{s+\lambda} \cdot h^*(s)
\right]
\nonumber
\\
&=
d^*(s) \cdot \frac{h^*(s)-h^*(s+\lambda)}{1-h^*(s+\lambda)}
+
w^{**}(h^*(s+\lambda),s)h^*(s+\lambda) h^*(s)
\left(
\frac{\lambda}{s+\lambda} 
- \frac{h^*(s)-h^*(s+\lambda)}{1-h^*(s+\lambda)}
\right).
\label{eq:a_peak-NP-MG1-wo-d}
\end{align}

Therefore, we can obtain (\ref{eq:a^*-NP-LCFS-MG1-wo-d})
with some calculations based on 
(\ref{eq:time-average-GIG1-LST}), (\ref{eq:lmd^dag-Ap-D}),
(\ref{eq:Pi(z,s)}), (\ref{eq:w**(z,s)}),
(\ref{eq:NP-LCFS-MG1-wo-d^*(s)}), and (\ref{eq:a_peak-NP-MG1-wo-d}).

\subsection{Derivation of (\ref{eq:a^*-NP-LCFS-GM1-wo-d})}
\label{appendix:NP-LCFS-GM1-w/o-d}

Let $L^{\rm A}$ denote the queue length seen by a randomly chosen
packet on arrival in steady state.  Since the non-preemptive LCFS
service discipline without discarding is work-conserving, the queue
length distribution immediately before arrivals is identical to that
in the FCFS service discipline, which is given by \cite[Page
251]{Klei75}
\[
\Pr(L^{\rm A} = k) = (1-\gamma) \gamma^k,
\quad
k=0,1,\ldots,
\]
where $\gamma$ denotes the unique solution of
(\ref{eq:gamma-def}). 

Let $W$ denote the waiting time of informative packets in steady state. 
Arriving packets become informative if (ii) they arrive at the empty
system or (ii) they arrive when a packet is being served and no
subsequent packets arrive in the exponentially distributed remaining
service time. Note here that for a randomly chosen arrival, the event
(i) happens with probability $1-\gamma$ and the event (ii) happens with
probability $\gamma(1-g^*(\mu))$. From these observations, we have
\begin{equation}
\Pr(W = 0) 
= \frac{1-\gamma}{1- \gamma g^*(\mu)},
\qquad
\Pr(W > 0) 
= \frac{\gamma (1-g^*(\mu))}{1- \gamma g^*(\mu)}.
\label{eq:W-pr-NP-LCFS-GM1-wo-d}
\end{equation}
With those, we obtain
\begin{align}
d^*(s) 
&=
\biggl[
\Pr(W=0) + \Pr(W >0) \cdot \frac{1}{1-g^*(\mu)} 
\int_0^{\infty} e^{-s x} \mu e^{-\mu x} (1-G(x))\dd x
\biggr] \frac{\mu}{s+\mu}
\nonumber
\\
&=
\biggl[
\Pr(W=0) 
+ \Pr(W>0)\cdot 
\frac{\mu(1-g^*(s+\mu))}{s+\mu}
\biggr]
\frac{\mu}{s+\mu}.
\label{eq:d(s)-NP-LCFS-GM1-wo-d}
\end{align}

By considering two exclusive cases $W_n = 0$ and $W_n > 0$,
we write the LST $a_{\peak}^*(s)$ of the peak
AoI $A_{\peak,n+1}$ in the form of (\ref{eq:a_peak-NP-GM1}).
When $W_n=0$, i.e., an informative packet finds the system empty on
arrival, the packet served next to this informative packet is also
informative and $W_{n+1}=0$. We then have from  (\ref{eq:NP-LCFS-A_peak-2}),
\begin{align}
a_{\peak,0}^*(s)
&=
\left[
g^*(s+\mu)\cdot\frac{\mu}{s+\mu} + (g^*(s)-g^*(s+\mu))\right]
\frac{\mu}{s+\mu}
\nonumber
\\
&=
\left[g^*(s) - \frac{sg^*(s+\mu)}{s+\mu}\right]
\frac{\mu}{s+\mu}.
\label{eq:NP-LCFS-GM1-peak-AoI0}
\end{align}

When $W_n > 0$, on the other hand, the number of non-informative packets
served before the arrival of the next informative packet need to be
taken into account. We define $b_k(s)$ ($s > 0$, $k=0,1,\ldots$)
as
\[
b_k^*(s) = \int_0^{\infty} e^{-s x} \cdot \frac{e^{-\mu x}(\mu x)^k}{k!} \dd G(x).
\]
If $L^{\rm A} = k-1$, the queue length immediately after the arrival
is equal to $k$.  Note here that
\[
\Pr(L^{\rm A} \geq k \mid L^{\rm A} > 0) = \gamma^{k-1}, 
\quad 
k=1,2,\ldots.
\]
We thus obtain from (\ref{eq:NP-LCFS-A_peak-2}),
\begin{align}
a_{\peak,+}^*(s)
&=
\sum_{k=1}^{\infty} \frac{b_k^*(s)}{1-g^*(\mu)}
\left[
1-\gamma^{k-1}
+ 
\gamma^{k-1} \cdot \frac{\mu}{s+\mu}\right]
\frac{\mu}{s+\mu}
\nonumber
\\
&=
\frac{1}{1-g^*(\mu)}
\biggl[
g^*(s)-g^*(s+\mu) 
- 
\frac{s[g^*(s+\mu-\mu \gamma) - g^*(s+\mu)]}{\gamma(s+\mu)}
\biggr] \frac{\mu}{s+\mu}.
\label{eq:NP-LCFS-GM1-peak-AoI+}
\end{align}
We can obtain (\ref{eq:a^*-NP-LCFS-GM1-wo-d})
with some calculations based on 
(\ref{eq:time-average-GIG1-LST}), (\ref{eq:lmd^dag-Ap-D}),
(\ref{eq:a_peak-NP-GM1}), (\ref{eq:W-pr-NP-LCFS-GM1-wo-d}),
(\ref{eq:d(s)-NP-LCFS-GM1-wo-d}), (\ref{eq:NP-LCFS-GM1-peak-AoI0}),
(\ref{eq:NP-LCFS-GM1-peak-AoI+}).

\section{Proof of Theorem \ref{theorem:LCFS-P-limit}}
\label{appendix:LCFS-P-limit}

We first consider the M/GI/1 queue. We rewrite (\ref{eq:mean-AoI-LCFS-P-MG1}) as
\[
\E[A] = \frac{1}{\lambda h^*(\lambda)}.
\]
Under the assumptions of Theorem \ref{theorem:LCFS-P-limit},
it follows from the initial value theorem \cite[Page 243]{Zemanian1965} 
that
\[
\lim_{\lambda \to \infty} \lambda h^*(\lambda) = h(0),
\]
which implies (\ref{eq:LCFS-P-limit}).

We then consider the D/GI/1 queue. Suppose that interarrival times
are given by $G_n = \tau$, where $\tau > 0$ is a real number. 
Because $\lambda = 1/\tau$, we consider the limit $\tau \to 0+$.
Obviously, the first two terms on the right-hand side of
(\ref{eq:PR-LCFS-GG1-EA}) converge to zero as $\tau \to 0+$. 
Furthermore, the third term is rewritten as
\begin{equation}
\frac{\zeta \E[G_{<H}]}{1-\zeta} 
= 
\frac{\tau\Pr(H > \tau)}{1-\Pr(H > \tau)}
=
\frac{\ds \int_{\tau}^\infty h(x) \dd x}{\ds \frac{1}{\tau}\int_0^{\tau} h(x)
\dd x}.
\label{eq:LCFS-P-DM1-limit-3}
\end{equation}
Because $h(x)$ is assumed to be continuous, we have
\begin{equation}
\lim_{\tau \to 0+} \frac{1}{\tau}\int_0^{\tau} h(x) \dd x = h(0),
\label{eq:LCFS-P-DM1-limit-3-h(0)}
\end{equation}
whether $h(0) \neq 0$ or not. Therefore, we obtain
(\ref{eq:LCFS-P-limit}) from (\ref{eq:LCFS-P-DM1-limit-3}),
(\ref{eq:LCFS-P-DM1-limit-3-h(0)}), and 
$\lim_{\tau \to 0+} \int_{\tau}^\infty h(x) \dd x = 1$.

\section{Proof of Theorem \ref{theorem:M/G/1-G/M/1-order}}
\label{appendix:M/G/1-G/M/1-order}

\subsection{Proof of the M/GI/1 case}

We first show $\E[A_{\rm LCFS}^{\rm NP-W/O}] \leq \E[A_{\rm FCFS}]$.
Because 
\[
1-\lambda x \leq 
e^{-\lambda x} \leq 
1 - \lambda x + \frac{\lambda^2 x^2}{2},
\quad
\lambda \geq 0,\,
x \geq 0, 
\]
we have
\begin{equation}
1-\rho
\leq
h^*(\lambda) 
\leq 
1 - \rho + \frac{\lambda^2 \E[H^2]}{2}.
\label{h^*(lambda)-bounds}
\end{equation}
It then follows from (\ref{eq:mean-age-MG1}),
(\ref{eq:EA-NP-LCFS-MG1-w/o-d}), and 
(\ref{h^*(lambda)-bounds}) that
\begin{align*}
\E[A_{\rm FCFS}] - \E[A_{\rm LCFS}^{\rm NP-W/O}]
&=
\frac{\lambda \rho \E[H^2]}{2(1-\rho)}
+
\frac{(1-\rho)\E[H]}{h^*(\lambda)} - \E[H]
\\
&=
\frac{\E[H]}{h^*(\lambda)}
\left(
\frac{\lambda^2\E[H^2]}{2} \cdot \frac{h^*(\lambda)}{1-\rho} + 1-\rho
- h^*(\lambda)
\right)
\\
&\geq
\frac{\E[H]}{h^*(\lambda)}
\left(
\frac{\lambda^2\E[H^2]}{2} + 1-\rho - h^*(\lambda)
\right)
\\
&\geq 0.
\end{align*}

Next we consider $\E[A_{\rm LCFS}^{\rm NP-W}] \leq \E[A_{\rm
LCFS}^{\rm NP-W/O}]$. It follows from Corollary
\ref{corollary:NP-LCFS} and
(\ref{eq:EA-NP-LCFS-MG1-w/o-d}) that
\begin{align*}
\E[A_{\rm LCFS}^{\rm NP-W/O}] - \E[A_{\rm LCFS}^{\rm NP-W}]
&=
\left(1 - \frac{1}{\rho + h^*(\lambda)}\right)
\left(\frac{\lambda\E[H^2]}{2} + \frac{h^*(\lambda)}{\lambda}
- h^{(1)}(\lambda)\right) 
+ \E[H] + \frac{(1-\rho)^2}{\lambda h^*(\lambda)}
- \frac{1}{\lambda}
\\
&=
\frac{h^*(\lambda)-(1-\rho)}
{\lambda h^*(\lambda) (\rho+h^*(\lambda))}
\biggl[
h^*(\lambda) 
\left(\frac{\lambda^2 \E[H^2]}{2} + h^*(\lambda)
- \lambda h^{(1)}(\lambda)\right) 
-(1-\rho)(\rho+h^*(\lambda))
\biggr]
\\
&=
\frac{h^*(\lambda)-(1-\rho)}
{\lambda h^*(\lambda) (\rho+h^*(\lambda))}
\biggl[
h^*(\lambda) 
\left(\frac{\lambda^2 \E[H^2]}{2} + h^*(\lambda)
- \lambda h^{(1)}(\lambda) -1\right) 
+ \rho [h^*(\lambda)-(1-\rho)]
\biggr], 
\end{align*}
which implies that 
\begin{align*}
\frac{\lambda^2 \E[H^2]}{2} + h^*(\lambda)
- \lambda h^{(1)}(\lambda) \geq 1
\ \Rightarrow\
\E[A_{\rm LCFS}^{\rm NP-W/O}] \geq \E[A_{\rm LCFS}^{\rm NP-W}].
\end{align*}
Note here that
\[
\frac{\lambda^2 \E[H^2]}{2} + h^*(\lambda)
- \lambda h^{(1)}(\lambda)
=
\int_0^{\infty} f(x) \dd H(x),
\]
where $f(x)$ ($x \geq 0$) is given by
\[
f(x) = \frac{\lambda^2 x^2}{2} + e^{-\lambda x} + \lambda x e^{-\lambda x}.
\]
It is readily seen that $f(x) \geq 1$ ($x \geq 0$) because 
$f(0)=1$ and for $x \geq 0$, 
\[
\frac{\dd}{\dd x} f(x) 
=
\lambda^2 x (1- e^{-\lambda x}) \geq 0.
\]
We thus have
\[
\int_0^{\infty} f(x) \dd H(x) 
\geq 
\int_0^{\infty} \dd H(x) =1,
\]
which completes the proof.

\subsection{Proof of the GI/M/1 case}

Note that $\E[A_{\rm LCFS}^{\rm P}] \leq \E[A_{\rm LCFS}^{\rm NP-W}]$
and $\E[A_{\rm LCFS}^{\rm NP-W/O}] \leq \E[A_{\rm FCFS}]$ are readily
verified from (\ref{eq:mean-age-GM1}), (\ref{eq:mean-AoI-GM1-LCFS-P}),
(\ref{eq:EA-NP-LCFS-GM1-with-d}), and
(\ref{eq:EA-NP-LCFS-GM1-w/o-d}).  We thus prove only $\E[A_{\rm
LCFS}^{\rm NP-W}] \leq \E[A_{\rm LCFS}^{\rm NP-W/O}]$ below.

By definition (\ref{eq:gamma-def}), $\gamma$ satisfies 
\[
\gamma = \int_0^{\infty}e^{-(\mu - \mu\gamma)x} \dd G(x)
=
\sum_{n=0}^{\infty} b_n \gamma^n,
\]
where $b_n$ ($n = 0,1,\ldots$) is given by
\[
b_n = \int_0^{\infty} \frac{e^{-\mu x}(\mu x)^n}{n!}\dd G(x).
\]
We then have
\[
\gamma \geq b_0 + b_1 \gamma,
\]
so that
\begin{equation}
\gamma \geq \frac{b_0}{1-b_1} 
= \frac{g^*(\mu)}{1 - \mu (-g^{(1)}(\mu))}.
\label{eq:gamma-bound}
\end{equation}
Similarly, we can verify that
\begin{align}
-g^{(1)}(\mu - \mu\gamma)
&= 
\int_0^{\infty}xe^{-(\mu - \mu\gamma)x} \dd G(x)
\nonumber
\\
&=
\sum_{n=0}^{\infty} 
\frac{(n+1)b_{n+1}}{\mu} 
\cdot
\gamma^n
\nonumber
\\
&\geq
\frac{b_1}{\mu} 
+
\frac{2b_2}{\mu} 
\cdot
\gamma
\nonumber
\\
&=
(-g^{(1)}(\mu)) + \gamma \mu g^{(2)}(\mu).
\label{eq:g(1)-gamma-bound}
\end{align}
It then follows from (\ref{eq:EA-NP-LCFS-GM1-with-d}),
(\ref{eq:EA-NP-LCFS-GM1-w/o-d}), (\ref{eq:gamma-bound}), and
(\ref{eq:g(1)-gamma-bound}) that
\begin{align*}
\E[A_{\rm LCFS}^{\rm NP-W}] 
&\leq
\E[H] + \frac{\E[G^2]}{2\E[G]} 
+ 
\rho \biggl(
(-g^{(1)}(\mu)) 
+
\gamma \mu g^{(2)}(\mu)
\biggr)
\\
&\leq
\E[H] + \frac{\E[G^2]}{2\E[G]} 
+ 
\rho \biggl(
-g^{(1)}(\mu - \mu\gamma)
\biggr)
\\
&=
\E[A_{\rm LCFS}^{\rm NP-W/O}],
\end{align*}
which completes the proof.

\section{Proof of Theorem \ref{theorem:MG1-FCFS<P-LCFS}}
\label{appendix:theorem:MG1-FCFS<P-LCFS}

We first consider (\ref{eq:theorem:MG1-FCFS<P-LCFS}).  
It follows from (\ref{eq:mean-age-MG1}),
(\ref{eq:mean-AoI-LCFS-P-MG1}), and (\ref{h^*(lambda)-bounds})
that
\begin{align*}
\frac{\E[A_{\rm LCFS}^{\rm P}] - \E[A_{\rm FCFS}]}{\E[H]}
&=
\frac{1}{h^*(\lambda)} 
-
\frac{\lambda \E[H^2]}{2(1-\rho)\E[H]} -1
\\
& \geq 
\frac{1}{1-\rho + \lambda^2 \E[H^2]/2}
-
\frac{\lambda \E[H^2]}{2(1-\rho)\E[H]} -1
\\
&=
\frac{2}{2(1-\rho) + \rho^2 ((\Cv[H])^2 + 1)}
-
\frac{\rho}{2(1-\rho)} ((\Cv[H])^2 + 1)-1,
\end{align*}
where we use $\E[H^2] = (\E[H])^2 ((\Cv[H])^2 + 1)$.
Straightforward calculations then yield
\begin{align}
\frac{\E[A_{\rm LCFS}^{\rm P}] - \E[A_{\rm FCFS}]}{\E[H]}
\geq
\frac{-[\rho^2 (\Cv[H])^2]^2 - 2 \rho^2 (\Cv[H])^2
+ \rho^2 \{(2-\rho)^2-2\}}
{2\rho (1-\rho) \{1+(1-\rho)^2 + \rho^2 (\Cv[H])^2\}}.
\label{eq:diff-P-FCFS-MD1}
\end{align}
Note that the denominator on the right-hand side of
(\ref{eq:diff-P-FCFS-MD1}) is always positive for $\rho \in (0,1)$.
We thus consider $f_1(x):= -x^2 - 2x + \rho^2 \{(2-\rho)^2-2\}$ ($x \geq
0$), which corresponds to the numerator with $x=\rho^2 (\Cv[H])^2 \geq
0$.  It is easy to verify that $f_1(x) \geq 0$ is
equivalent to
\[
(2-\rho)^2-2 \geq 0\ 
\mbox{and}\ 
x \leq - 1 + \sqrt{1 + \rho^2 \{(2-\rho)^2-2\}},
\]
from which (\ref{eq:theorem:MG1-FCFS<P-LCFS}) follows. It can also be 
verified that for $\rho \in (0, 2-\sqrt{2})$, $v(\rho)$ is a
decreasing function of $\rho$ and $\lim_{\rho \to 0+} v(\rho) = 1$.

\end{document}